\documentclass{svproc}

\usepackage{algorithmic}
\usepackage{algorithm}
\usepackage{graphicx}
\usepackage{cite}
\usepackage{color}
\usepackage{amssymb}
\usepackage{psfrag}
\usepackage{cleveref}
\usepackage[normalem]{ulem}
\usepackage{ifthen, latexsym} %
\usepackage{amsmath}
\usepackage{cases}
\usepackage{array}
\usepackage[skip=7pt]{caption}
\usepackage{caption}
\usepackage{subcaption}
\usepackage{enumerate}

\def\notA{\overline{A}}

\def\notS{\overline{S}}

\def\Z{\mathbb{Z}}
\def\SR{\mathcal{SR}}
\def\MSR{\mathcal{MSR}}
\def\MASR{\mathcal{MASR}}

\def\clus{\mathcal{C}}

\newcommand{\Pee}[0]{\ensuremath{{\mathbb P}}}
\newcommand{\Ee}[0]{\ensuremath{{\mathbb E}}}

\newcolumntype{L}{>{\centering\arraybackslash}m{2.8cm}}

\author{Deniz Ozsoyeller}

\institute{Department of Computer Engineering,\\Final International University, North Cyprus, Turkey,\\\email{deniz.ozsoyeller@final.edu.tr} }

\begin{document}

\title{\LARGE \bf Multi-robot Symmetric Rendezvous Search on the Line\\ with an Unknown Initial Distance}
\maketitle

\begin{abstract}
In this paper, we study the symmetric rendezvous search problem on the line with
$n > 2$ robots that are unaware of their locations and the
initial distances between them. In the symmetric version of this
problem, the robots execute the same strategy.
The multi-robot symmetric rendezvous algorithm, $\MSR$ presented in this paper is an
extension our symmetric rendezvous algorithm, $\SR$ presented in~\cite{ozsoyeller2013symmetric}.
We study both the synchronous and asynchronous cases of the problem. The asynchronous version of $\MSR$ algorithm is
called $\MASR$ algorithm. We consider that robots start executing $\MASR$ at different times.
We perform the theoretical analysis of $\MSR$ and $\MASR$ and show that their competitive ratios are $O(n^{0.67})$ and $O(n^{1.5})$, respectively.
Finally, we confirm our theoretical results through simulations. 
\end{abstract}

\section{Introduction}
\label{sec:intro}

In the rendezvous search problem, two or more players that are
unaware of their locations in the environment and cannot communicate over long distances want to meet as quickly as possible.
This problem arises when two people become separated shopping in a mall, when two parachutists who have to meet after a simultaneous landing in a large field, or when rescuers search for a lost hiker who wants to be found. As well as its obvious connection with real life problems, this interesting problem also has various applications in robotic search-and-rescue, network formation, multi-robot exploration and mapping.

In robotic search-and-rescue, rescuers (robots) can search for victims and survivors in urban disasters and explosions. Multiple robots could also be employed to explore and build the map of unknown environments such as mine fields, contaminated areas or distant planets that can be hazardous or inaccessible to humans. To accomplish this task, they might rendezvous to collaboratively explore the environment. Suppose multiple robots are initially employed to perform surveillance in a large environment. Upon detection of an event, they may have to form a network to propagate information as quickly as possible. When robots with limited communication range and different unknown locations are dispersed in a large environment, network formation problem becomes closely related to rendezvous search problems.

The rendezvous search problem has two different versions, depending on
whether or not the robots can meet in advance of the search to agree on
the strategies each will execute. In asymmetric rendezvous
search, the robots can meet in advance and choose distinct strategies. For example, one can wait while the other carries out an exhaustive search. This is different to symmetric rendezvous search, where the robots execute the same strategy, since they do not have chance to agree on their roles. In this
version, it is not necessary to implement a different strategy on each robot. Therefore, it is appealing for robotics applications.

Let $x_i$ denote the initial location of robot $i \geq 1$ in an
environment $Q$ and $d(x_1,...,x_n)$ denote the minimum possible distance traveled before rendezvous.
The efficiency of a rendezvous strategy ${\cal S}$ is often
measured by its competitive ratio
\begin{align}
\max_{x_1,...,x_n \in Q} \frac{S_1(x_1,...,x_n)+...+S_n(x_1,...,x_n)}{d(x_1,...,x_n)}
\label{eq:introcomprat}
\end{align}
where $S_i(x_1,...,x_n)$ denotes the (expected) distance traveled by robot $i$ before rendezvous.
The competitive ratio of ${\cal S}$ is the worst case deviation
of the performance of ${\cal S}$ from this optimal behavior. A strategy is said to be {\em competitive}
if its competitive ratio is a constant.

The contributions of this paper are as follows. We study the symmetric rendezvous search problem with multi-robots on the line for an unknown initial distance. Moreover, the robots do not know their positions or directions. We first present a symmetric rendezvous algorithm $\MSR$ for the synchronous setting of the problem. $\MSR$ is an extension of our algorithm $\SR$ presented in \cite{ozsoyeller2013symmetric}. We perform the theoretical analysis of $\MSR$ and show that its competitive complexity is $O(n^{0.67})$. Second, we study the problem in the asynchronous setting. For this setting of the problem, $\MSR$ is called $\MASR$. We prove that $\MASR$ has a competitive complexity of $O(n^{1.5})$. Finally, we verify the theoretical results that are obtained for both cases in simulations.

The paper is organized as follows. We present an overview
of related work in Section~\ref{sec:related}. $\SR$ algorithm is
introduced in Section~\ref{sec:preliminary}. We formulate the multi-robot rendezvous
search problem and present $\MSR$ Algorithm in Section~\ref{sec:formulation}. In Section~\ref{sec:msranalysis}, we perform the analysis of $\MSR$.
We present the asynchronous case of the problem in Section~\ref{sec:asynchintro} and perform the analysis of this case in Section~\ref{sec:masranalysis}.
We present the simulation results in Section~\ref{sec:simulations}.
Finally, we provide concluding remarks in Section~\ref{sec:conclusion}. 
\section{Related Work}
\label{sec:related}

The rendezvous search problem can be formulated in various environments such as line, plane, circle(ring) or graph. In this paper, $n > 2$ robots are placed on a line with an unknown initial distance between them. A road, a street, a river, a corridor, a railway can me modeled as a line. The rendezvous search problem on the line is studied both for the symmetric \cite{alpern1995rendezvous, anderson1995rendezvous, baston1999note, uthaisombut2006symmetric, han2007improved, stachowiak2009asynchronous} and asymmetric \cite{baston1998rendezvous, gal1999rendezvous, alpern1999asymmetric, alpern2000pure} versions. Many previous studies focus on the asymmetric version with the players who know their initial distance or its distribution  \cite{alpern2000pure,stachowiak2009asynchronous}. However, the problem has not been well studied for the symmetric case and unknown initial distance. In the previous version of this paper \cite{ozsoyeller2013symmetric}, we present a new symmetric rendezvous algorithm for two robots that has a competitive ratio of 17.686 for total distance traveled and a competitive ratio of 24.843 for total time. Both are improvements over the algorithm of Baston and Gal \cite{baston1998rendezvous}, in which the distance distribution is not known and has a competitive ratio of 26.650. In this paper, we extend our work \cite{ozsoyeller2013symmetric} to multi-robots and provide the theoretical and simulations results for both the synchronous and asynchronous cases of the problem.

Lim et. al \cite{lim1997rendezvous} studies the rendezvous of $m \leq n$ blind, speed one, players. The players are placed by a random permutation onto the integers 1 to $n$ on the line. Each player points randomly to the right or left, thus have no common notion of a positive direction on the line. The initial distance between each player is known and equal to 1. The least expected rendezvous time of $m$ players is given by $R^{a}_{n, m}$ and $R^{s}_{n, m}$ for the asymmetric and symmetric strategy, respectively.
$R^{a}_{3, 2}$ is 47/48 and $R^{s}_{n, n}$ is asymptotic to $n/2$. Prior to this study, Alpern and Lim \cite{lim1996minimax} focus on the asymmetric version of the same problem and minimizing the maximum time to rendezvous rather than the expected time. The asymmetric
value of the $n$-player minimax rendezvous time $M_n$ has an upper bound $n/2 + (n/ \log n) + o(n/ \log n)$. Gal \cite{gal1999rendezvous} presents a simpler strategy for the problem in \cite{lim1996minimax} and shows that the worst case meeting time has an asymptotic behavior of $n/2 + O(\log n)$.

The asynchronous case of the rendezvous search problem has not received as much attention on the line \cite{marco2006asynchronous, stachowiak2009asynchronous}, as in graphs \cite{marco2006asynchronous, dessmark2006deterministic, thomas2001many, kowalski2005polynomial, chalopin2007rendezvous, czyzowicz2010meet, dieudonne2013meet} and in geometric environments \cite{bampas2010almost, czyzowicz2010asynchronous, czyzowicz2010meet, dieudonne2014price, collins2010tell}. We aim to fill this gap in our work. Although \cite{marco2006asynchronous} concentrates on the asynchronous rendezvous in graphs, the authors also present a deterministic rendezvous algorithm for two agents located on an infinite line. They think of an adversary that interferes the starting times and the motion of an agent. If the agents execute the same deterministic algorithm and the adversary makes them move in the same direction at the same speed, then they will never meet. Thus, the agents have distinct identifiers, called labels.
Labels are two different nonempty binary strings, and each agent knows its own label.
Based on its label, each agent produces the label $L^{*}$.
This bit string is a motion pattern which consists of three consecutive segments and is followed
by the agent. Because of the asynchronous setting, at the time $t$ when agent $X$ completes the second segment of the $p$-th bit, agent $Y$ can be
already executing the $p$-th bit.
The cost of their algorithm is $O(D|L_{min}|^{2})$ when $D$ is known and $O((D+|L_{max}|)^3)$ when $D$ is unknown. Here, $|L_{min}|$ and $|L_{max}|$ denote the lengths of the shorter and longer label of
the agents, respectively. This bound is improved to $O(Dlog^{2}D+DlogD|L_{max}|+D|L_{min}|^{2}+|L_{max}||L_{min}|log |L_{min}|)$ by Stachowiak \cite{stachowiak2009asynchronous}. Thomas and Pikounis \cite{thomas2001many} study the multi-player rendezvous search on a complete graph.
The paper focuses on whether players should stick together or split up and meet again later
when some but not all of them meet. Authors show that among the class of strategies that require no memory
and are stationary, sticking together is the optimal strategy. However, split up and meet again strategy achieves faster expected rendezvous times in most situations.

In the robotics literature, there are two types of rendezvous
problems. The first type is interested in robot tracking and navigation
toward a moving object (target) where the agents can
observe each other’s state. The second type is the rendezvous {\em search} problem which we study in this paper. The first type of the problem focuses on the control-theoretic aspects which include combining the kinematics equations of the robot and the target. The target can be another mobile robot, a satellite, a moving convoy or a human. The main difference between these types is that the rendezvous search problem does not use the state information. The robots are not equipped with a (long range) sensory system. Therefore, they cannot determine the position of the other robot to achieve rendezvous. The robots do not necessarily know their current location. Moreover, they do not (and cannot) know the initial distance or direction to the other robot.

\section{Preliminaries}
\label{sec:preliminary}

In this section, we briefly explain $\SR$ Algorithm. The extension of this algorithm to multi-robots is introduced in Section~\ref{sec:formulation}. In the earlier version of the problem, two robots are placed on a line with an with an unknown initial distance between them. The initial distance between the robots is set to $2d$, where $d = r^{k+\delta}$, for $\delta\in(0,1]$ and $k\in\Z^+$. To each robot, a non-negative sequence $f_{-1}, f_0, f_1, f_2, \ldots$ is assigned, where $f_{-1}=0$ and $$f_i = r^{i+\epsilon} \mbox{ for round } i \geq 0.$$
Here, $r>0$ is the expansion radius and $\epsilon \in (0,1]$ is a uniformly distributed random variable. The robots use the same expansion radius $r$. They choose their $\epsilon$ values independently at the start of the algorithm and use them throughout the algorithm. They start executing the algorithm at the same time and continue to synchronize their movements with waiting times.

The algorithm proceeds in rounds indexed by integers $i \geq 0$. If the robots choose the same direction at the beginning and stick with these directions in later rounds, they will never meet. Thus, randomization is used to break the symmetry between the robots. In round $i$, the robot flips a coin to determine its itinerary. Each round is divided into two phases: phase-1 and phase-2.

We now describe the movement of robot-1, who starts at x = 0. In the $i$th round, the robot starts at one of $x=\pm f_{2i-1}$, each with probability $1/2$. If the robot tosses heads, then it follows \emph{Right-Wait-Left-Wait} motion pattern; it moves right to the point $x=f_{2i}$ in phase-1, waits for some time at the end of this phase, then it moves left to the point $x=-f_{2i+1}$ and waits for some time at the end of this phase. If the robot tosses tails, then it follows \emph{Left-Wait-Right-Wait} motion pattern; it moves left to $x=-f_{2i}$ in phase-1, waits for some time at the end of this phase, then moves right to $x=f_{2i+1}$ in phase-2 and waits for some time at the end of this phase. A robot determines its waiting time at the end of each phase of a round considering the possible total distance traveled in that phase and assuming that the other robot is using $\epsilon = 1$. At the end of an unsuccessful round $i \geq 0$, the possible  configurations of the robots are ($\pm f_{2i+1}, 2d\pm f_{2i+1}$). This is also the initial configuration for round $i+1$.  
\section{Problem Formulation}
\label{sec:formulation}

In this section, we present the extension of the $\SR$ algorithm to multi-robots. $n >$ 2 robots are placed on a line with equal initial distances between them. The robots do not know the initial distances between each other. Two robots are adjacent to each other if there is no robot located between them. For example, in Fig.~\ref{fig:msrdeployment}, robot-3 is adjacent to both robot-2 and robot-4, and robot-1 is only adjacent to robot-2. As in algorithm $\SR$, the initial distance between two adjacent robots is set to 2$d$. Robot-$j$ is initially located at $x = (j-1)2d$, where integer $j\in [1, n]$. Let the expansion radius $r >$ 1 be fixed. We determine the choice for $r$ in Section~\ref{sec:msranalysis} for the synchronous case and in Section~\ref{sec:masranalysis} for the asynchronous case of the problem.

\begin{figure}
\psfrag{2d}{\small{2$d$}}
\psfrag{4d}{\small{4$d$}}
\psfrag{6d}{\small{6$d$}}
\psfrag{f2i}{\tiny{$f_{2i}$}}
\psfrag{2df2i}{\tiny{2$d$+$f_{2i}$}}
\psfrag{4df2i}{\tiny{4$d$+$f_{2i}$}}
\psfrag{6dmf2i}{\tiny{6$d$-$f_{2i}$}}
\psfrag{c1}{\small{$\clus_{1} = \{1\}$}}
\psfrag{c2}{\small{$\clus_{2} = \{2\}$}}
\psfrag{c3}{\small{$\clus_{3} = \{3\}$}}
\psfrag{c4}{\small{$\clus_{4} = \{4\}$}}
\centering
\includegraphics[width=0.9\columnwidth]{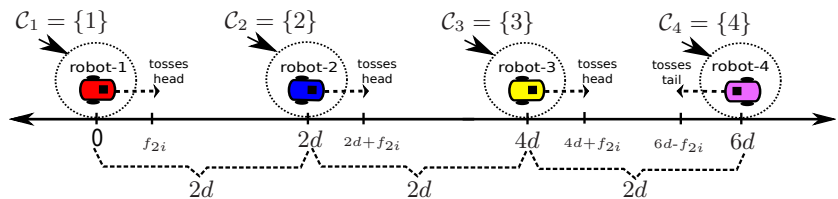}
\caption{Initial configurations of four robots on the line. 2$d$ is the initial distance between adjacent robots. Left/Right arrows show the directions that the robots choose in round $i$ as a result of their coin flips.}
\label{fig:msrdeployment}
\end{figure}

We call the multi-robot version of the $\SR$ algorithm, $\MSR$. Each robot independently executes $\MSR$ algorithm without $\epsilon$ value. Thus, for each robot
$$f_i = r^{i} \mbox{ for round } i \geq 0.$$
When two robots meet in round $i$, they rendezvous into a cluster and the robot with the smaller id becomes the leader of the cluster. The robots inside the cluster thereafter sticks together and follow the motion pattern that is determined by the coin flip of their leader. At the beginning of a round, only the cluster leader flips a coin. Clusters can meet moving towards each other and cannot meet if they move in tandem. More than one cluster can meet in a round.

$\clus_{j}$ represents the set of the robots in a cluster which is indexed by $j$ with respect to its position from the left on the line. We denote the leader robot in cluster $\clus_{j}$ by $L_{j}$ and the initial location of $\clus_{j}$ by $I_{j}$. Let $\clus^{*}$ be the number of clusters at the beginning of round $i$. $\clus^{*} = n$ in round $i = 0$ and decreases by one whenever two clusters meet into a new cluster. At the beginning of round $i = 0$, $\clus_{j} = \{j\}$, thus $\left\vert{\clus_{j}}\right\vert = 1$. Rendezvous occurs in round $i$, when $\clus^{*} = 1$.

Fig.~\ref{fig:msrexecution} shows sample executions of $\MSR$ algorithm when $n=5$ and $n=7$. In the top and bottom plots, the rendezvous occurs in 6 rounds, while in middle plot, it occurs in 9 rounds. Thus, the robot travels the maximum distance in the middle plot. In simulations (Section~\ref{sec:simulations}), we observe that the distance traveled by the robot is proportional to the number of rounds. Given that the rendezvous occurs in round $i$, the distance traveled by the clusters in that round is maximized when Robot-1 and Robot-$n$ are the leaders of the last two clusters on the line. Such a case occurs in the second execution of $\MSR$, where robot-1 and robot-6 are the last clusters to meet. We now explain the execution of algorithm $\MSR$ shown in the top plot of Fig.~\ref{fig:msrexecution}. Here, Robot-4 and Robot-5 meet in round 3. At the beginning of round 4, the clusters are $\clus_{1}=\{1\}$, $\clus_{2}=\{2\}$, $\clus_{3}=\{3\}$ and $\clus_{4} = \{4,5\}$ with the leader robots $L_1 = 1$, $L_2 = 2$, $L_3 = 3$ and $L_4 = 4$. In round 5, Robot-1 first meets Robot-2, then Robot-3. At the beginning of round 6, the clusters are $\clus_{1} = \{1, 2, 3\}$ and $\clus_{2} = \{4,5\}$ with the leader robots $L_1 = 1$ and $L_2 = 4$, respectively. Rendezvous occurs in round 6, when $\clus_{1}$ and $\clus_{2}$ meet.

\begin{figure}
\centering
\includegraphics [width=\columnwidth, height=5.5cm]{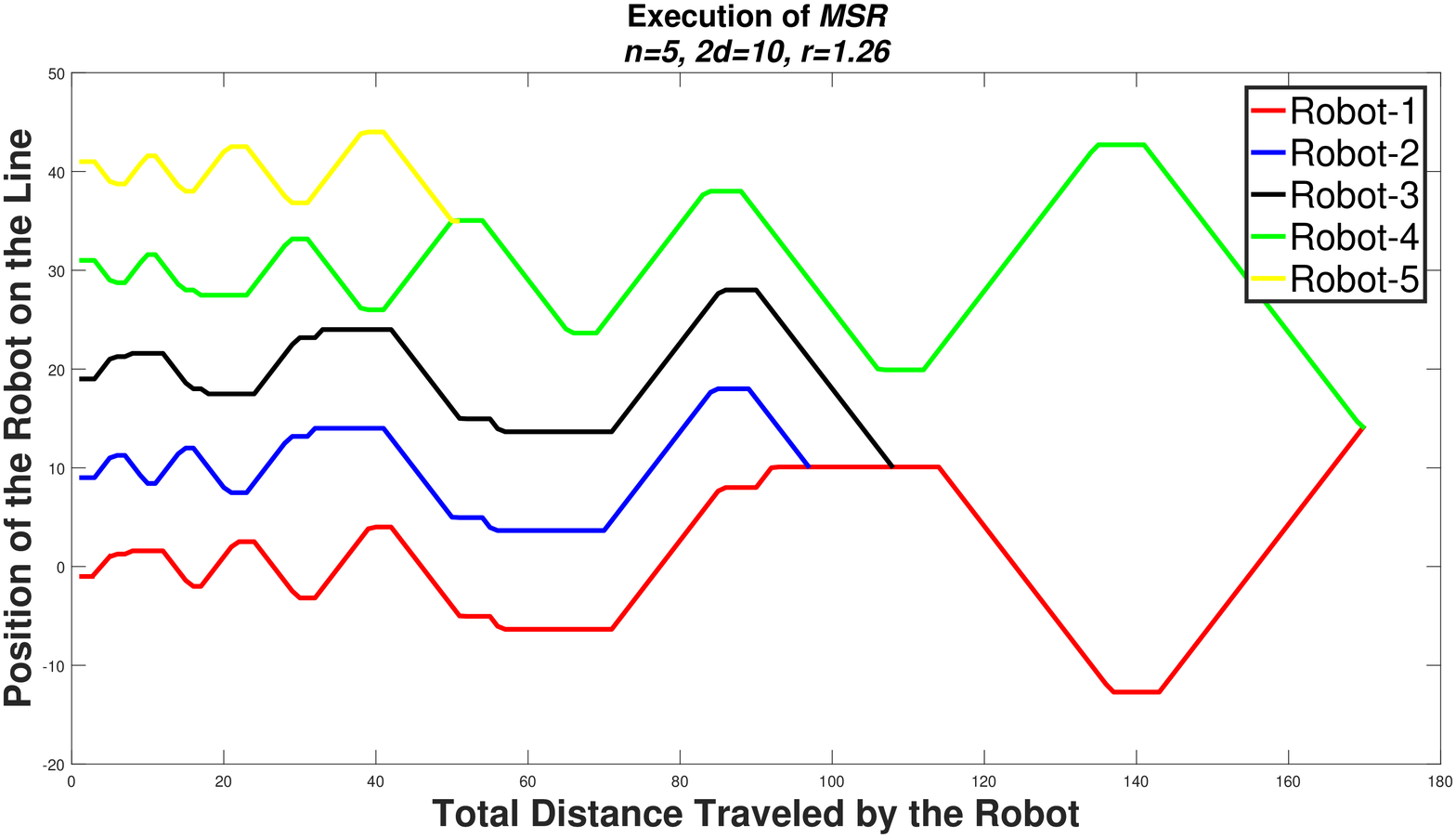}\\\vspace{0.2cm}
\includegraphics [width=\columnwidth, height=5.5cm]{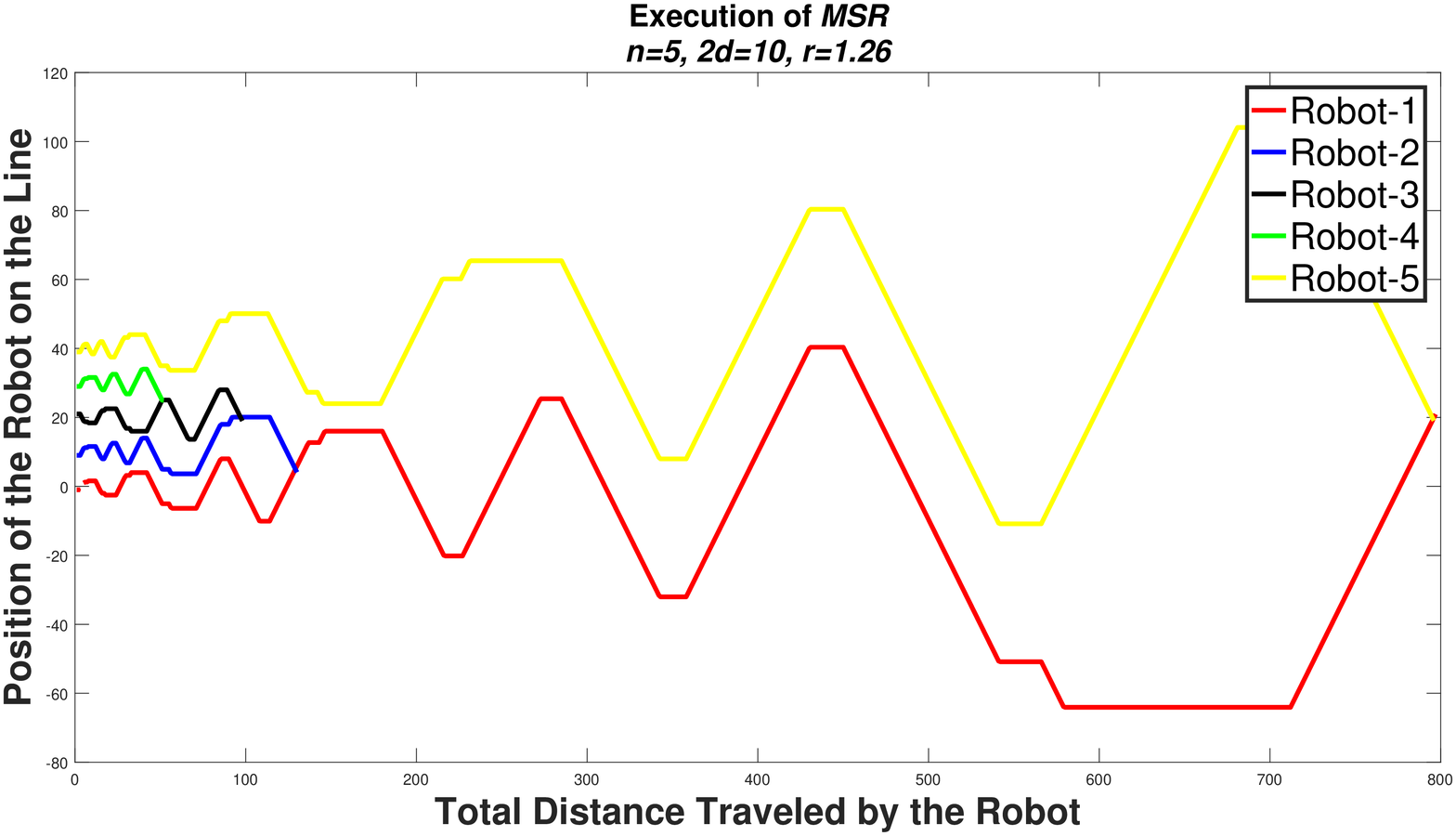}\\\vspace{0.2cm}
\includegraphics [width=\columnwidth, height=5.5cm]{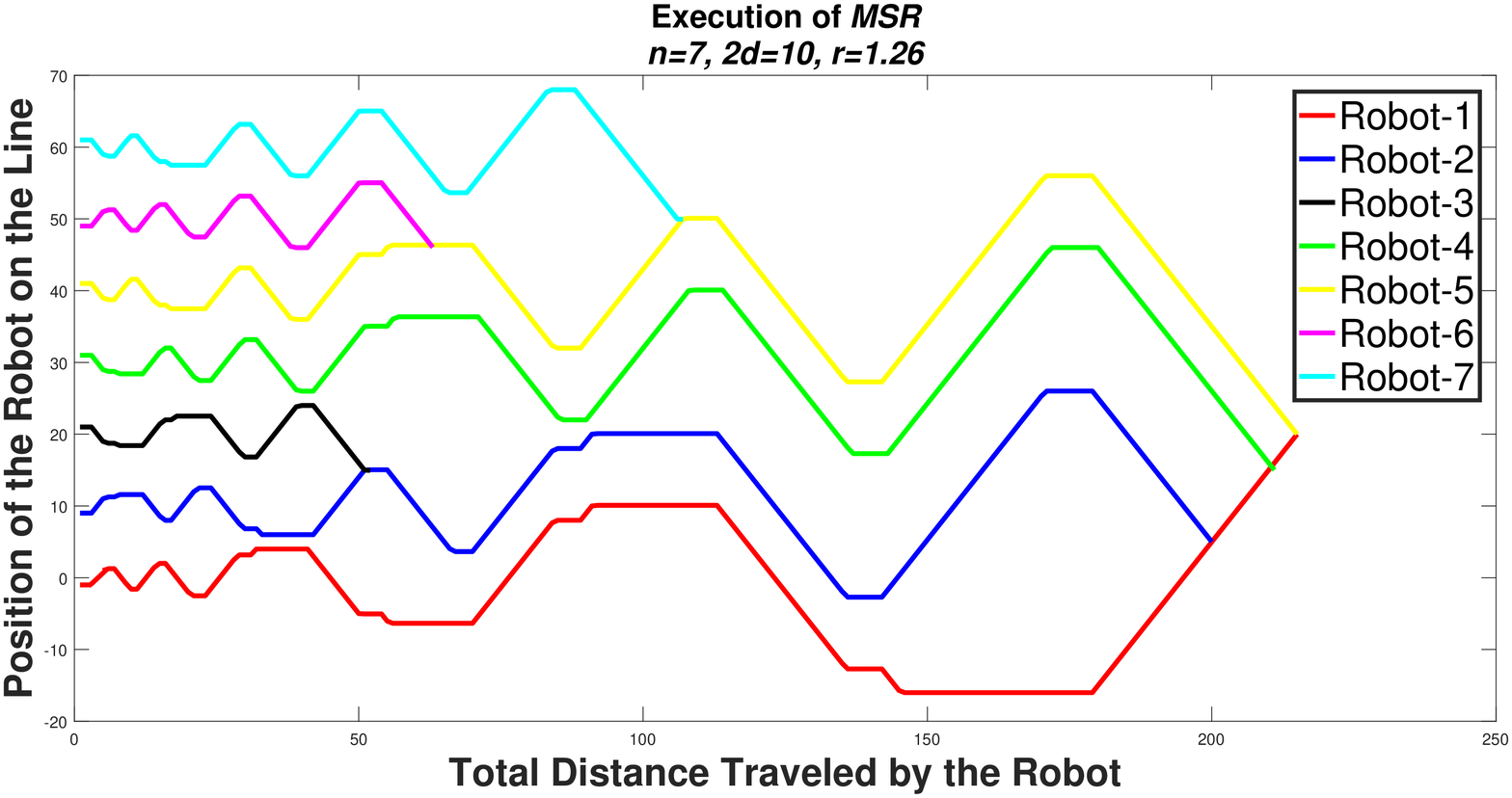}
\caption{Sample executions of Algorithm $\MSR$.}
\label{fig:msrexecution}
\end{figure} 
\section{Analysis of $\MSR$ Algorithm}
\label{sec:msranalysis}

In this section, we analyze the performance of $\MSR$ algorithm and find an upper bound on the expected distance traveled by the robot. Note that due to the symmetric strategies, the performance all the robots are the same. For omniscient robots, the best offline algorithm would be for them to move toward each other and meet at $x = (n-1)d$. Thus, our competitive ratio will be calculated in comparison with distance $(n-1)d$.

We denote the probability of a cluster getting a head in one flip of a fair coin by $p$. Let the random variable $X$ follow the binomial distribution with parameters $C^{*}$ and $p$, then the probability of getting exactly $k^{*}$ heads in $\clus^{*}$ coin flips is given by
\begin{align*}
\Pee[k^{*};\clus^{*},p]&=\Pee[X = k^{*}] = \binom{\clus^{*}}{k^{*}}\left(p\right)^{k^{*}}\left(1-p\right)^{\clus^{*}-k^{*}} = \binom{\clus^{*}}{k^{*}}\left({\frac{1}{2}}\right)^{\clus^{*}}
\end{align*}

For $\lceil \frac{\clus^{*}}{2}\rceil + 1 < \clus^{*}$ and $\lfloor \frac{\clus^{*}}{2}\rfloor - 1 > 0$, let $h_{1}, h_{2}, h_{3}$, and $h_{4}$ denote the events $X = \lfloor\frac{\clus^{*}}{2}\rfloor -1$, $X = \lfloor\frac{\clus^{*}}{2}\rfloor$, $X =\lceil\frac{\clus^{*}}{2}\rceil$ and $X =\lceil \frac{\clus^{*}}{2}\rceil+1$ in round $i$, respectively. Event
\[
 H_i =
  \begin{cases}
   h_{1} \vee h_{3} \vee h_{4}              & \text{if $\clus^{*}$ is even,}  \\
   h_{1} \vee h_{2} \vee h_{3} \vee h_{4}   & \text{if $\clus^{*}$ is odd.}
  \end{cases}
\]
The probability of event $H_i$ is then given by
\begin{numcases}
{\Pee[H_i] =}
\Pee[h_{1} \vee h_{3} \vee h_{4}]              & if $\clus^{*}$ is even, \label{evscase1}
\\
\Pee[h_{1} \vee h_{2} \vee h_{3} \vee h_{4}]   & if $\clus^{*}$ is odd. \label{evscase2}
\end{numcases}
For both cases (\ref{evscase1}) and (\ref{evscase2}), $\Pee[H_i] \geq 1/2$. Thus, we consider that $\Pee[H_i] = 1/2$.

Let $S^{*}_i$ be the event that all the robots rendezvous into one cluster in round $i$, i.e., $\clus^{*} = 1$. Assuming that the algorithm is still active in round $i$, let $A_{i}^{j}$ be the event that cluster $\clus_{j}$ initially moves to the right and $\notA_{i}^{j}$ be the event that cluster $\clus_{j}$ initially moves to the left in round $i$. Adjacent clusters can meet if event
$$E_1 = \bigg(A_{i}^{j} \wedge \notA_{i}^{j+1}\bigg) \text{ or } E_2 = \bigg(\notA_{i}^{j} \wedge A_{i}^{j+1}\bigg)$$
occurs, and cannot meet if event
$$E_3 = \bigg(A_{i}^{j} \wedge A_{i}^{j+1}\bigg) \text{ or } E_4 = \bigg(\notA_{i}^{j} \wedge \notA_{i}^{j+1}\bigg)$$
occurs. Let $\alpha=k/2+1.5\log n$ and $\clus^{*}_{U} = \lfloor \frac{\clus^{*}}{2}\rfloor$.

In the following lemma, we establish the relation between $k^{*}$ and $\clus^{*}$.
\begin{lemma}
For $r = 1.26$, if $1 \leq k^{*} \leq \clus^{*}_{U}$ heads are obtained in $\clus^{*} \geq \frac{n}{2^{i-\alpha}}$ coin flips in round $i \geq \alpha$, then $\clus^{*}$ decreases by $k^{*}$ at the end of this round.
\label{lemma:half}
\end{lemma}
\begin{proof}
The maximum possible distance between clusters $\clus_{j}$ and $\clus_{j+k^{*}}$ is given by
\begin{align}
max(dist(\clus_{j}, \clus_{j+k^{*}}))=(n-\clus^{*}+k^{*})2d.
\label{eq:maxd}
\end{align}
At the end of round $i$, $C^{*}$ is minimized and (\ref{eq:maxd}) is maximized when $k^{*} = \clus^{*}_{U}$. Therefore, we substitute $\clus^{*}_{U}$ with $k^{*}$ in (\ref{eq:maxd}) to obtain
\begin{align}
max\left(dist\left(\clus_{j}, \clus_{j+\clus^{*}_{U}}\right)\right) &=\left(n-\clus^{*}+\clus^{*}_{U}\right)2d =
\left(n - \frac{n}{2^{i-\alpha+1}}\right)2d.
\label{eq:maxd2}
\end{align}
Adjacent clusters can only meet if event $\left(E_1 \vee E_2\right)$ occurs. That is, if their coin flips are different from each other.
The worst-case scenario is defined by the event $E_{5}$ which occurs when $k^{*} = \clus^{*}_{U}$ and all $k^{*}$ heads in the sequence of $\clus^{*}$ coin flips appear consecutively. Note that the order of the clusters' coin flips matters, but the order of clusters' ids does not matter. We use the clusters' ids in order only for the ease of representation.

We prove this lemma by showing that $\clus^{*}$ decreases by $k^{*}$ even when event $E_5$ occurs. Given $E_5$ occurs, there is a sequence of events
\begin{align}
A_{i}^{j}, \notA_{i}^{j+1}, ... , \notA_{i}^{j+k^{*}},\label{eq:seqevents}
\end{align}
such that $(\clus_j, \clus_{j+1})$ are the only adjacent clusters with different outcomes of coin flips. Therefore, $\clus^{*}$ decreases by $k^{*}$ at the end of round $i$ only if $\clus_{j}$ meets the next $k^{*} \leq  \clus^{*}_{U}$ clusters in this sequence. In this case, the distance traveled by the cluster is maximized for $k^{*} = \clus^{*}_{U}$. Consider the example sequences of coin flips; H,H,H,H,T,T,T,T, and T,T,T,T,H,H,H,H, where $n = \clus^{*} = 8$ and $k^{*} =  \clus^{*}_{U} = \frac{\clus^{*}}{2} = 4$. Here, $\clus^{*}$ decreases by $k^{*}$ at the end of round $i$ only when $\clus_{4}$ meets all the clusters in the subsequences T,T,T,T and H,H,H,H in the first and second sequences, respectively. Assuming that event $E_5$ occurs, $\clus_{j}$ meets the next $k^{*}$ clusters in (\ref{eq:seqevents}), if
\begin{align}
f_{2i} \geq \frac{max(dist(\clus_{j}, \clus_{j+k^{*}}))}{2}.
\label{eq:gmaxd2}
\end{align}
Since (\ref{eq:gmaxd2}) holds true for $r = 1.26$ and $1 \leq k^{*} \leq  \clus^{*}_{U}$, $\clus^{*}$ decreases by $k^{*}$.

Although, it seems from (\ref{eq:seqevents}) like it is enough to have $k^{*} = 1$ head in round $i$ to decrease $\clus^{*}$ by $k^{*} =  \clus^{*}_{U}$, this may not be always true. For example, consider the sequence T,T,T,T,T,T,T,H in round $i$, where $n = \clus^{*} = 8$. In round $i$, only $\clus_{7} =\{7\}$ and $\clus_{8} = \{8\}$ can meet. When these clusters stick together in phase-2 of round $i$, we have $\clus_{7} = \{7,8\}$ and $L_{7} = 7$. $\clus_{7}$ then continues moving right, following the direction of $L_{7}$. This results in all the clusters on the line to move in the same direction till the end of this round. Therefore, $\clus^{*}$ decreases by $k^{*}=1$.
\end{proof}

\begin{lemma}
For $r = 1.26$, if $ \clus^{*}_{U} < k^{*} < \clus^{*}$ heads are obtained in $\clus^{*}$ coin flips in round $i \geq \alpha$, then $\clus^{*}$ decreases by $\clus^{*}-k^{*}$ at the end of this round.
\label{lemma:morehalf}
\end{lemma}
\begin{proof}
Since $\clus^{*}_{U} < k^{*} < \clus^{*}$, the proof is the same as Lemma~\ref{lemma:half} when $\clus^{*}-k^{*}$ heads are obtained. Therefore, $\clus^{*}$ decreases by $\clus^{*}-k^{*}$ at the end of this round.
\end{proof}

Let $S_i$ denote the event that $\clus^{*}$ decreases by at least $\clus^{*}_{U}- 1$ at the end of round $i$. We say that a round is successful if event $S$ occurs in that round, unsuccessful otherwise. From Lemmas \ref{lemma:half} and \ref{lemma:morehalf}, we conclude that if event $H_i$ occurs in round $i$, then event $S_i$ also occurs. Therefore, the probability of round $i$ being successful is given by $\Pee[S_i] = 1/2$. The minimum number of rounds required for the rendezvous in some round $i'$ is achieved if $S_i$ occurs in each round $\alpha \leq i\leq i'$. This number is maximized when event $h_{1}$ occurs and is given by the recursive function
$$T(n) = T(\lceil n/2 \rceil) + 1.$$
Substituting $\log n$\footnote{The logarithms in this paper are binary logarithms.} into the recurrence yields
\begin{align}
T(n) = O\left(\log n\right).
\label{eq:numrounds}
\end{align}

Let $R_i$ be the event that the algorithm is still active in round $i$. It follows from (\ref{eq:numrounds}) that if event $S_i$ occurs less than $\log n$ times in $i-\alpha$ rounds, then event $S^{*}$ cannot occur, thus $\Pee[R_i] = 1$. The probability of $R_i$ is given by
\begin{align}
&\Pee[R_i] \notag\\
&\quad =\sum_{x = 0}^{\log n-1} \binom{i-\alpha}{x} {\bigg(\Pee[S_i]\bigg)}^{x} {\bigg(1-\Pee[S_i]\bigg)}^{i-\alpha-x}\notag\\
&\quad =\sum_{x = 0}^{\log n-1} \binom{i-\alpha}{x} {\bigg(\frac{1}{2}\bigg)}^{x} {\bigg(\frac{1}{2}\bigg)}^{i-\alpha-x}\notag\\
&\quad \leq {\bigg(\frac{1}{2}\bigg)}^{i-\alpha}\sum_{x = 0}^{\log n - 1}\frac{(i-\alpha)^{x}}{x!}\label{eq:prri}.
\end{align}
Using the finite taylor series polynomial approximation \cite{taylorseries}, (\ref{eq:prri}) yields
\begin{align}
\Pee[R_i] \leq\frac{(i-\alpha)^{\log n}}{2^{i-\alpha}(\log n)!}.\label{eq:prri2}
\end{align}

We divide the execution of Algorithm $\MSR$ into three stages. \emph{Stage-1} consists of rounds $0 \leq i \leq \alpha-1$ that adjacent clusters do not travel far enough to meet. The first round in which the adjacent clusters might meet is round $\alpha$. \emph{Stage-2} consists of the rounds $\alpha \leq i \leq \alpha + \log n - 1$. The first round that rendezvous can occur is the ($k/2+\alpha + \log n)$th round. Stage-3 consists of rounds $i \geq k/2+\alpha + \log n$. We now study the three stages of Algorithm $\MSR$ for the synchronous case of the problem. Sections~\ref{subsec:stg1}-\ref{subsec:stg3} present the distance traveled analysis of stages 1-3, respectively.
\subsection{\textbf{Analysis of Stage-1}}
\label{subsec:stg1}
This section presents the computation of the expected distance traveled during Stage-1.
\begin{lemma}
\label{lemma:stage1} The expected distance traveled during Stage-1 satisfies
\begin{align}
&\sum_{i=0}^{\alpha-1} \Ee[D_i \mid R_i] \Pee[R_i] < (r+2)\frac{nr^{k}}{r^2-1}.
\label{eq:stage1}
\end{align}
\end{lemma}
\begin{proof}
Adjacent clusters cannot meet when round $i < \alpha$. Therefore, $\Pee[R_i] = 1$ in this stage. The possible itineraries of adjacent clusters based on their
initial configurations in round $i$ are shown in Fig.~\ref{motionpat}.

\begin{figure}
\psfrag{Clusterj}[bl][c][0.8][0]{\tiny{$\clus_{j}$}}
\psfrag{Clusterj1}[bl][c][0.8][0]{\tiny{$\clus_{j+1}$}}
\psfrag{L1}[bl][c][0.8][90]{\tiny{$I_j$}}
\psfrag{L2}[bl][c][0.8][90]{\tiny{$I_{j+1}$}}
\psfrag{A}[bl][c][0.8][0]{\tiny{$X_{\clus_{j}}^{i}=I_j+f_{2i-1}$}}
\psfrag{B}[bl][c][0.8][0]{\tiny{$Y_{\clus_{j}}^{i}=I_j+f_{2i}$}}
\psfrag{C}[bl][c][0.8][0]{\tiny{$Z_{\clus_{j}}^{i}=I_j+f_{2i+1}$}}
\psfrag{-A}[bl][c][0.8][0]{\tiny{$-X_{\clus_{j}}^{i}=I_j-f_{2i-1}$}}
\psfrag{-B}[bl][c][0.8][0]{\tiny{$-Y_{\clus_{j}}^{i}=I_j-f_{2i}$}}
\psfrag{-C}[bl][c][0.8][0]{\tiny{$-Z_{\clus_{j}}^{i}=I_j-f_{2i+1}$}}
\psfrag{A2}[bl][c][0.8][0]{\tiny{$X_{\clus_{j+1}}^{i}=I_{j+1}+f_{2i-1}$}}
\psfrag{B2}[bl][c][0.8][0]{\tiny{$Y_{\clus_{j+1}}^{i}=I_{j+1}+f_{2i}$}}
\psfrag{C2}[bl][c][0.8][0]{\tiny{$Z_{\clus_{j+1}}^{i}=I_{j+1}+f_{2i+1}$}}
\psfrag{-A2}[bl][c][0.8][0]{\tiny{$-X_{\clus_{j+1}}^{i}=I_{j+1}-f_{2i-1}$}}
\psfrag{-B2}[bl][c][0.8][0]{\tiny{$-Y_{\clus_{j+1}}^{i}=I_{j+1}-f_{2i}$}}
\psfrag{-C2}[bl][c][0.8][0]{\tiny{$-Z_{\clus_{j+1}}^{i}=I_{j+1}-f_{2i+1}$}}
\centering
\includegraphics[width=\textwidth]{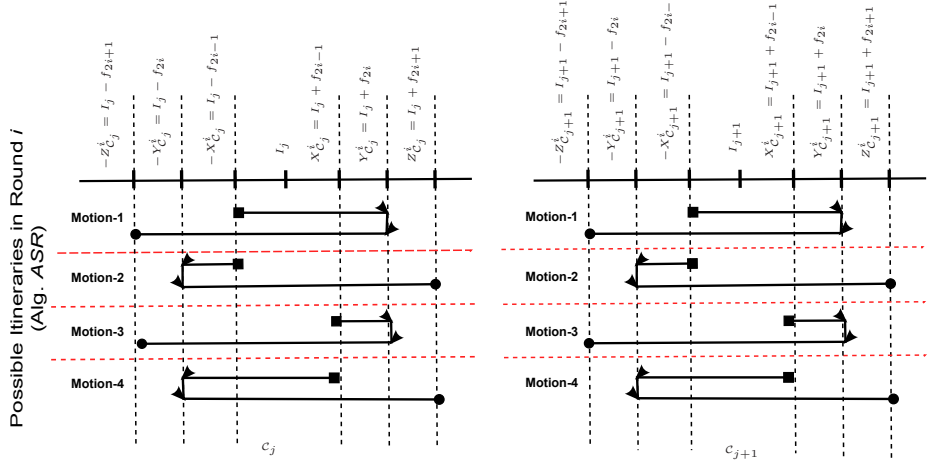} 
\caption{The possible itineraries of the adjacent clusters $\clus_j$ and $\clus_{j+1}$ executing algorithm $\MSR$. The arrows show the direction that the cluster is moving in round $i$.}
\label{motionpat}
\end{figure}
The distance traveled (the length of an itinerary) by a cluster in an unsuccessful round $i$ is either $D_i = f_{2i+1} + 2f_{2i} - f_{2i-1}$ or $D_{i} = f_{2i+1} + 2f_{2i} + f_{2i-1}$, each with equal probability. Therefore, we have
\begin{align}
\Ee[D_i \mid \notS_i] = \Ee[D_i \mid {\notS_i}^{*} ] &= \Ee[f_{2i+1} + 2f_{2i} \mid {\notS_i}^{*} ] \notag\\
&= (r^{2i+1} + 2r^{2i}).\label{eq:unsuccessful}
\end{align}
Using (\ref{eq:unsuccessful}), we obtain
\begin{align*}
&\sum_{i=0}^{\alpha-1} \Ee[D_i \mid R_i] \Pee[R_i] = \sum_{i=0}^{\alpha-1} \Ee[D_i \mid {\notS_i}^{*}] \Pee[{\notS_i}^{*}] \\
&\quad = \sum_{i=0}^{\alpha-1} \left(f_{2i+1} + 2f_{2i}\right)\cdot 1 = \sum_{i=0}^{\alpha-1} (r^{2i+1} + 2r^{2i})\\
&\quad = (r+2)\frac{r^{2(k/2+1.5\log_{2} n)} - 1}{r^2-1}\\
&\quad < (r+2)\frac{r^{k+3\log_{2} n}}{r^2-1} \\
&\quad < (r+2)\frac{nr^{k}}{r^2-1} \text{ for $r$ = 1.26.}
\end{align*}
\end{proof}

\subsection{\textbf{Analysis of Stage-2}}
\label{subsec:stg2}
This section presents the computation of the expected distance traveled during Stage-2 which encompasses the rounds $\alpha \leq i < \alpha + \log n - 1$. Adjacent clusters can meet in this stage. However, rendezvous cannot occur until round $i =\alpha + \log n$. Thus, the algorithm is still active during this stage, i.e. $\Pee[R_{i}] = 1$.
\begin{lemma}
\label{lemma:stage2}
The expected distance traveled during Stage-2 satisfies
\begin{align}
&\sum_{i=\alpha}^{\alpha + \log n - 1} \Ee[D_i \mid R_i] \Pee[R_i] < (r+2)\frac{n^{1.67}r^{k}}{r^2-1}.
\label{eq:stage2}
\end{align}
\end{lemma}
\begin{proof}
The expected distance traveled by a cluster in this stage is given by
\begin{align*}
&\sum_{i=\alpha}^{\alpha+\log n-1} \Ee[D_i \mid R_i] \Pee[R_i] = \sum_{i=\alpha}^{\alpha+\log n - 1}  \Ee[D_i \mid {S_i}^{*}] \Pee[{\notS_i}^{*}] \notag\\
&\qquad = \sum_{i=\alpha}^{\alpha+\log n - 1} \left(f_{2i+1} + 2f_{2i}\right) = \sum_{i=\alpha}^{\alpha+\log n - 1} \left(r^{2i+1} + 2r^{2i}\right) \notag\\
&\qquad < (r+2)\bigg(\frac{r^{2(k/2+2.5\log n)}- r^{2(k/2+1.5 \log n)}}{r^2-1}\bigg)\notag\\
&\qquad < (r+2)\frac{r^{k+5\log n}}{r^2-1} \notag\\
&\qquad < (r+2)\frac{n^{1.67}r^{k}}{r^2-1} \text{ for $r = 1.26$.}\notag
\end{align*}
\end{proof} 
\subsection{\textbf{Analysis of Stage-3}}
\label{subsec:stg3}

We compute the expected distance traveled for all rounds $i \geq (k/2 + \alpha + \log n )$. Unlike Stage-1 and Stage-2, rendezvous occurs in this stage with nonzero probability.

\begin{lemma}
\label{lemma:stage3} The expected distance traveled during Stage-3 satisfies
\begin{align}
&\sum_{i=\alpha+\log n}^{\infty} \bigg[(\Ee[ D_i \mid \notS^{*}_i]\Pee[\notS^{*}_i]) + (\Ee [ D_i \mid S^{*}_i ]\Pee[S^{*}_i])\bigg]\Pee[R_i]\notag\\
&\qquad<\frac{2n^{0.67}r^{k}(r+2)}{(2-r^2)(\log n)!}.
\label{eq:stage3}
\end{align}
\end{lemma}

\begin{proof}
Given ${S_i}^{*}$ holds, the distance traveled by the clusters is maximized when $\clus_{1}$ and $\clus_{2}$ where $L_1 = 1$ and $L_2=N$ are the last clusters to rendezvous. Thus, substituting $k^{*} = 1$ and $\clus^{*} = 2$ in (\ref{eq:maxd}), we have
\begin{align}
max(d(\clus_{j}, \clus_{j+1})) &= max(d(\clus_{1}, \clus_{2})) = d(I_1, I_2) = (n-1)2d.\label{eq:maxdlasttwo}
\end{align}
In this case, the four equiprobable initial configurations of the clusters are $(I_{1}\pm f_{2i-1}, I_{2}\pm f_{2i-1}) = (\pm f_{2i-1}, (n-1)2d\pm f_{2i-1})$. Thus, the expected distance traveled given $S^{*}_{i}$ holds, is
\begin{align}
&\Ee[D_{i} \mid S^{*}_{i}]\notag\\
&\quad = \frac{1}{4}\Ee[D_{i} \mid A_{i}^{j} \wedge \notA_{i}^{j+1}] + \frac{1}{4}\Ee[D_{i} \mid  \notA_{i}^{j} \wedge A_{i}^{j+1}]\notag\\
&\quad = \frac{1}{4}\Big(2f_{2i} + d(n-1)\Big) + \frac{1}{4}d(n-1)\notag\\
&\quad = \frac{1}{2}f_{2i} + \frac{d}{2}(n-1)\notag\\
&\quad = \frac{r^{2i}}{2}+\frac{d}{2}(n-1).
\label{eq:successful}
\end{align}
Comparing $\Ee[D_i\mid\notS^{*}_i ]$ and $\Ee[D_{i}\mid S^{*}_{i}]$ using (\ref{eq:unsuccessful}) and (\ref{eq:successful}), respectively, we have
$$\Ee[D_{i}\mid S^{*}_{i}] < \Ee[D_i\mid\notS^{*}_i ].$$
Thus, for the simplicity of subsequent computations, we assume that
$$\Ee[D_{i}\mid S^{*}_{i}] = \Ee[D_i\mid\notS^{*}_i ] = (r^{2i+1} + 2r^{2i}).$$
The expected distance traveled in Stage-3 is given by
\begin{align}
& \sum_{i=\alpha+\log n}^{\infty} \bigg[(\Ee[ D_i \mid \notS^{*}_i]\Pee[\notS^{*}_i]) + (\Ee [ D_i \mid S^{*}_i ]\Pee[S^{*}_i])\bigg]\Pee[R_i]\notag\\
& \qquad = \sum_{i=\alpha+\log n}^{\infty} \left(r^{2i+1} + 2r^{2i}\right)\frac{(i-\alpha)^{\log n}}{2^{i-\alpha}(\log n)!}\notag\\
&\qquad = \frac{r+2}{(\log n)!}\sum_{i=0}^{\infty} \frac{r^{2\left(i+\frac{k}{2}+2.5\log n\right)} \bigg(i+\log n\bigg)^{\log n}}{2^{i+\log n}}\notag\\
&\qquad < \frac{(r+2)r^{k+5\log n}}{n(\log n)!} \sum_{i=0}^{\infty} \bigg(\frac{r^{2}}{2}\bigg)^{i}i^{\log n}.\label{eq:stage3last}
\end{align}
We bound the infinite summation in (\ref{eq:stage3last}) by
$$\sum_{i=0}^{\infty} \bigg(\frac{r^{2}}{2}\bigg)^{i}i^{\log n} = \Theta\bigg(\frac{2}{2-r^2}\bigg(\frac{r^{3}}{2}\bigg)^{\log n}\bigg).$$
to obtain
$$<\frac{2n^{0.67}r^{k}(r+2)}{(2-r^2)(\log n)!} \text{ for $r$ = 1.26.}$$
\end{proof}

\begin{theorem}
For the choice of $r = 1.26$, $\MSR$ algorithm has a competitive ratio of $O(n^{0.67})$.
\label{theo:SRmulti}
\end{theorem}
\begin{proof}
The expected distance traveled is obtained by adding the expressions in equations (\ref{eq:stage1}), (\ref{eq:stage2}) and (\ref{eq:stage3}). Recalling that the initial distance between the adjacent clusters is $2d$, where $d = r^{k+\delta}$, we first replace each occurrence of $r^{k}$ with $dr^{-\delta}$. Then, we divide by $(n-1)d$ which is the length of the optimal offline path between the clusters. This expression is maximized at $\delta = 0$. In turn, the choice of $r = 1.26$ gives the competitive ratio guarantee of $O(n^{0.67})$.
\end{proof} 
\section{$\MASR$ Algorithm}
\label{sec:asynchintro}

Until now, we assume that the robots start executing the $\MSR$ algorithm at the same time. Although this is a standard assumption, it may often be unrealistic: robots may be created in different parts of the environment modeled as a line, oblivious to each other. Hence, in this section we investigate the symmetric rendezvous of robots that start searching at different times.

Recall that robots wait at the end of each stage of a round to keep their motions synchronized. For the asynchronous setting, we do not use idle times introduced in $\MSR$. The resulting algorithm is called $\MASR$. We consider that robot-$j$ starts executing algorithm $\MASR$ $t_{j}$ time late, where $t_{j}$ is a random integer value drawn from a discrete uniform distribution over the interval $(0, (n-1)2d)$. Robots are unaware of each other's latency.

\begin{proposition}
\label{ch5prop:ahead}
Consider two clusters $\clus_{1}$ and $\clus_{2}$, where $L_1 = 1$ (robot-1) and $L_2 = 2$ (robot-2). Assume that these clusters do not meet any other clusters before round $i$. Let $t_1 < t_2$. Depending on the values of $t_1$, $t_{2}$, and the coin flips of $L_1$ and $L_2$, $L_2$ can reach round $i$ earlier than $L_1$ despite of its late start.
\end{proposition}

\begin{proof}
Let $D_i$ be the distance traveled by the robot in an unsuccessful round $i$. Since $D_i$ depends on the coin flips of the robot in round $i$ and $i -1$, it can vary among the robots. Without the idle times, total time $T_i$ in round $i$ is $T_{i} = D_{i}$. In contrast to the algorithm $\MSR$, $T_{i}$ of the robots in the algorithm $\MASR$ can be different from each other. For the following case, we find out the robot that is first to reach round $i$: Consider that none of the outcomes of the two consecutive coin flips of $L_2$ until round $i$ are the same. Further consider that all the outcomes of the coin flips of $L_1$ until round $i$ are the same. In this case, $T^{2}_{i} < T^{1}_{i}$. If $t_1$ and $t_2$ have values such that the inequality
$$T^{2}_{i} + t_2 < T^{1}_{i} + t_1$$
holds, then $L_2$ arrives round $i$ earlier than $L_1$.
\end{proof}

In $\MASR$, the robots do not start each phase of a round at the same time.
Therefore, when one robot starts phase-1(2) of round $i$, another robot can be moving in phase-2(1) of the same or another round.
Moreover, before the robot reaches its destination in a phase, the other robot can finish its current round, flip a coin to start a new round and change its direction. If adjacent clusters $\clus_{j}$ and $\clus_{j+1}$ meet, $\clus_{j+1}$ becomes the leader if $\clus_{j}$ is executing a smaller round than $\clus_{j+1}$. In the next section, we perform the analysis of $\MASR$ without the knowledge of $t_{j}$.
\section{Analysis of the $\MASR$ Algorithm}
\label{sec:masranalysis}

Recall that the initial location of a cluster is $I_j = (j-1)2d$.
We use the following variables in the analysis: $t^{*}$ is the time $\clus_j$ arrives $x = I_{j+1}$; $i^{*}$ is the round that $\clus_j$ is
executing at $t^{*}$; $j^{*}$ is the round that $\clus_{j+1}$ is executing at $t^{*}$.
Let $\Delta = \left|i^{*}-j^{*}\right|$.
\begin{lemma}
\label{lemma:alwaysmeetrightof2d}
If $\clus_{j+1}$ is moving on the left side of $x = I_{j+1}$ at $t^{*}$, then the rendezvous has already occurred before $t^{*}$.
\end{lemma}
\begin{proof}
Since $I_{j} < I_{j+1}$, if $\clus{j+1}$ is moving on the left side of $x = I_{j+1}$ when $\clus{j}$ arrives $x = I_{j+1}$, then
$\clus_{j}$ should have already met $\clus{j+1}$ on its way to $x = I_{j+1}$.
\end{proof}
We assume that adjacent clusters cannot meet before $t^{*}$. Therefore, we can conclude from Lemma~\ref{lemma:alwaysmeetrightof2d} that
$\clus{j+1}$ is in moving on the right of $x = I_{j+1}$ at $t^{*}$. Let events $E_{1}^{*}$-$E_{4}^{*}$ correspond to the events $E_{1}$-$E_{4}$, respectively, when round $i = i^{*}$ for $\clus_j$, and $i = j^{*}$ for $\clus_{j+1}$.
The possible destinations $dest_{j}$ and $dest_{j+1}$ of $\clus_j$ and $\clus_{j+1}$ at $t^{*}$ are given by the states
\begin{align*}
& s_1 = (Y_{\clus_{j}}^{i^{*}}, Z_{\clus_{j+1}}^{j^{*}}) = (I_{j} + f_{2i^{*}}, I_{j+1}+f_{2j^{*}+1}), \\
& s_2 = (Z_{\clus_{j}}^{i^{*}}, Y_{\clus_{j+1}}^{j^{*}}) = (I_{j} + f_{2i^{*}+1}, I_{j+1}+f_{2j^{*}}), \\
& s_3 = (Y_{\clus_{j}}^{i^{*}}, Y_{\clus_{j+1}}^{j^{*}}) = (I_{j} + f_{2i^{*}}, I_{j+1} +f_{2j^{*}}) \text{, and}  \\
& s_4 = (Z_{\clus_{j}}^{i^{*}}, Z_{\clus_{j+1}}^{j^{*}}) = (I_{j} + f_{2i^{*}+1}, I_{j+1}+f_{2j^{*}+1}),
\end{align*}
that correspond to the events $E_{1}^{*}$-$E_{4}^{*}$, respectively.
The coin flips of clusters are independent from each other, thus each event occurs with the probability of 1/4.
Let $\alpha^{*} = k/2+2.75\log n+3$. In Lemmas~\ref{lemma:asynchE1}-\ref{lemma:halfasynch}, we consider that round $i = i^{*} \geq \alpha^{*}$.
We next study the possible rendezvous conditions at $t^{*}$.
\begin{figure}
\centering
\psfrag{head}{$\clus_{j+1}$ flips head at $Z_{\clus_{j+1}}^{j^{*}}$}
\psfrag{twod}{\small{$I_{j+1}$}}
\psfrag{Cj}{\small{$\clus_{j}$}}
\psfrag{Cj1}{\small{$\clus_{j+1}$}}
\psfrag{A}{$X_{\clus_{j+1}}^{j^{*}}$}
\psfrag{B}{\small{$Y_{\clus_{j+1}}^{j^{*}}$}}
\psfrag{C}{\small{$Z_{\clus_{j+1}}^{j^{*}}$}}
\psfrag{D}{\small{$Y_{\clus_{j+1}}^{j^{*}+1}$}}
\includegraphics[width=0.8\columnwidth]{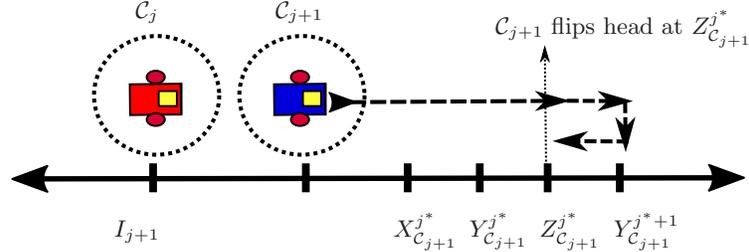}
\caption{The worst-case scenario at $t^{*}$, which occurs when $dest_{j+1} = Z_{\clus_{j+1}}^{j^{*}}$, where $\clus_{j+1}$ flips a coin to
start round $j^{*}+1$.}
\label{worstitinerary}
\end{figure}
\begin{lemma}
Regardless of the value of $\Delta$, adjacent clusters always meet when event $E_{1}^{*}$ occurs.
\label{lemma:asynchE1}
\end{lemma}
\begin{proof}
Recall that event $E_{1}^{*}$ corresponds to state $s_1 = (Y_{\clus_{j}}^{i^{*}}, Z_{\clus_{j+1}}^{j^{*}})$.
We study the the rendezvous behavior of adjacent clusters before $t^{*}$ on the left side of $I_{j}$ when $\clus_{j+1}$
is executing phase-1 of round $j^{*}$. Thus, we have $dest_{j} = -X_{C_{j}}^{i^{*}}$ and $dest_{j+1} = -Y_{C_{j+1}}^{j^{*}}$.
$\clus_{j}$ should have flipped head, if its destination is $Y_{\clus_{j}}^{i^{*}}$ when event $E_{1}^{*}$ occurs. Therefore,
when $\clus_{j}$ arrives $dest_{j} = -X_{C_{j}}^{i^{*}}$, it starts moving right towards $\clus_{j+1}$. $\clus_{j}$ and $\clus_{j+1}$ always meet when event $E_{2}^{*}$ occurs if
\begin{align}
dest_{j+1} \geq dest_{j} &\Leftrightarrow -Y_{C_{j+1}}^{j^{*}} \leq -X_{C_{j}}^{i^{*}} \notag\\
&\Leftrightarrow f_{2j^{*}} \geq \left(I_{j+1}-I_{j}\right) + f_{2i^{*}-1}\notag\\
&\Leftrightarrow r^{2j^{*}} \geq max\left(dist\left(\clus_j, \clus_{j+\clus^{*}_{U}}\right)\right) + r^{2(j^{*}-\Delta)-1}. \label{eq:compstate1}
\end{align}
When $\Delta \geq 0$, (\ref{eq:compstate1}) is true for the choice of $r = 1.26$.
\end{proof}

\begin{lemma}
Regardless of the value of $\Delta$, adjacent clusters always meet when event $E_{2}^{*}$ occurs.
\label{lemma:asynchE2}
\end{lemma}
\begin{proof}
Since $E_{2}^{*}$ corresponds to state $s_2 = (Z_{\clus_{j}}^{i^{*}}, Y_{\clus_{j+1}}^{j^{*}})$,
$\clus_{j+1}$ starts moving left towards $\clus_{j}$ when it arrives $dest_{j+1}$.
Thus, $\clus_{j}$ and $\clus_{j+1}$ always meet when event $E_{2}^{*}$ occurs if
\begin{align}
dest_{j} \geq dest_{j+1} &\Leftrightarrow Z_{C_{j}}^{i^{*}} \geq Y_{C_{j+1}}^{j^{*}}  \notag\\
&\Leftrightarrow I_{j} + f_{2i^{*}+1} \geq I_{j+1}+f_{2j^{*}} \notag\\
&\Leftrightarrow f_{2i^{*}+1} \geq \left(I_{j+1}-I_{j}\right) + f_{2j^{*}} \notag\\
&\Leftrightarrow r^{2i^{*}+1} \geq max\left(dist\left(\clus_j, \clus_{j+\clus^{*}_{U}}\right)\right) + r^{2(i^{*}-\Delta)},
\label{eq:compstate2}
\end{align}
which is true for $\Delta \geq 0$ and the choice of $r = 1.26$.
\end{proof}

\begin{lemma}
When $\Delta \geq 1$, adjacent clusters always meet when event $E_{3}^{*}$ occurs.
\label{lemma:asynchE3}
\end{lemma}
\begin{proof}
When event $E_{3}^{*}$ occurs, $\clus_j$ and $\clus_{j+1}$ can meet only if $\clus_j+1$ catches up with $\clus_j$ before it starts moving left to execute phase-2 of round $i^{*}$. This is given by
\begin{align}
dest_{j} \geq dest_{j+1} &\Leftrightarrow Y_{C_{j}}^{i^{*}} \geq Y_{\clus_{j+1}}^{j^{*}} \notag\\
&\Leftrightarrow I_{j} + f_{2i^{*}} \geq I_{j+1}+f_{2j^{*}}\notag\\
&\Leftrightarrow f_{2i^{*}} \geq \left(I_{j+1}-I_{j}\right) +f_{2j^{*}} \notag\\
&\Leftrightarrow r^{2i^{*}} \geq max\left(dist\left(\clus_j, \clus_{j+\clus^{*}_{U}}\right)\right) + r^{2(i^{*}-\Delta)}\label{eq:compstate3}.
\end{align}
From (\ref{eq:compstate1}), we can derive the conclusion that (\ref{eq:compstate3}) is true for $\Delta \geq 1$ and $r = 1.26$.
\end{proof}

\begin{lemma}
When $\Delta \geq 1$, adjacent clusters always meet when event $E_{4}^{*}$ occurs.
\label{lemma:asynchE4}
\end{lemma}
\begin{proof}
The worst-case scenario when event $E_{4}^{*}$ occurs would be $\clus_{j+1}$ to arrive its destination $Z_{\clus_{j+1}}^{j^{*}}$ before $\clus_{j}$, then flip a coin at this position to start round $j^{*}+1$, and start moving right to its new destination $Y_{\clus_{j+1}}^{j^{*}+1}$.
Fig.~\ref{worstitinerary} shows this scenario. In this case, adjacent clusters meet if
\begin{align}
dest_{j} \geq dest_{j+1} &\Leftrightarrow  Z_{C_{j}}^{i^{*}} \geq Y_{\clus_{j+1}}^{j^{*}+1} \notag\\
&\Leftrightarrow I_{j} + f_{2i^{*}+1} \geq I_{j+1}+f_{2j^{*}+2} \notag\\
&\Leftrightarrow f_{2i^{*}+1} \geq \left(I_{j+1}-I_{j}\right) +f_{2j^{*}+2} \notag\\
&\Leftrightarrow r^{2i^{*}+1} \geq max\left(dist\left(\clus_j, \clus_{j+\clus^{*}_{U}}\right)\right) + r^{2(i^{*}-\Delta)+2} \label{eq:compstate4}.
\end{align}
(\ref{eq:compstate4}) with $\Delta = 1$ is the same as (\ref{eq:compstate2}) with $\Delta = 0$. Thus, it is true for the choice of $\Delta \geq 1$ and $r = 1.26$.
\end{proof}

\begin{lemma}
The probability of $S_i$ for the asynchronous case is the same as the synchronous case of the problem, which is given by $\Pee[S_i] = 1/2$.
\label{lemma:halfasynch}
\end{lemma}
\begin{proof}
Consider that $\clus_j$ is executing round $i^{*}$, which we assume to be the highest round on the line.
Let $\Delta^{*} = \{...,\Delta_{\clus_j, \clus_{j+1}},...,\Delta_{\clus_{j}, \clus_{\clus^{*}}}\}$ be the set of $\Delta$ value between $\clus_j$ and all the other clusters. Let $P(x)$ denote the statement $x \geq 1$ and
$Q(x)$ denote the statement $x = 0$. We study the proof in three cases.

\emph{Case (i):} The statement $\forall x\in\Delta^{*}P(x)$ is true. In this case, we assume that whenever $\clus_{j}$ reaches the initial location of another cluster $j'$, $\Delta_{j,j'} \geq 1$. Since $i^{*} \geq \alpha^{*}$, $f_{2i} \geq (n-1)2d$, for $r \geq 1.26$.
This implies that $\clus_{j}$ passes from the initial locations of all the other clusters on the line. From Lemmas~\ref{lemma:asynchE1}-\ref{lemma:asynchE4}, we prove that when $\Delta \geq 1$, adjacent clusters can meet in all the possible
events $E_{1}^{*}$-$E_{4}^{*}$. Therefore, regardless of the value of $k^{*}$, $\clus_{j}$ always meets at least $\clus^{*}_{U}$ clusters in round $i^{*}$. As a result, $\Pee[S_{i}] = 1$.

\emph{Case (ii):} The statement $\forall x\in\Delta^{*}Q(x)$ is true. It implies that whenever $\clus_{j}$ reaches the initial location of another cluster $j'$, $\Delta_{j,j'} = 0$. Since all the clusters on the line are executing the same round, this case is similar to the synchronous case. Thus, we next show that Lemma~\ref{lemma:half} holds for Case (ii). In Lemmas~\ref{lemma:asynchE1} and \ref{lemma:asynchE2}, we prove that adjacent clusters can meet if event $E_{1}^{*}$ or $E_{2}^{*}$ occurs in round $i^{*}$. Moreover, we show that a cluster can travel far enough to meet the next $\clus^{*}_{U}$ clusters. This completes the proof in Lemma~\ref{lemma:half}. Thus, $\Pee[S_{i}] = 1/2$.

\emph{Case (iii):} The statement $\forall x,y\in\Delta^{*}\big(x \neq y \wedge |\{x|P(x)\}| \neq |\{y|Q(y)\}|\big)$ is true.
We compare the asynchronous and synchronous cases when there are $\clus^{*}$ clusters on the line in both of them.
Consider that while all the clusters executing $\MSR$ are in round $i^{*}$, there is at least one cluster executing $\MASR$ in round $i^{*}$. We match the coin flips of the same indexed clusters executing $\MSR$ and $\MASR$ in round $i^{*}$. Let $m$ be the total number of matches. In round $i^{*}$ of the synchronous case, $k^{*}$ heads are obtained in $\clus^{*}$ coin flips. Let $k^{*}_{m}$ be the number of heads from $k^{*}$ heads that appear in $m$ matches. We split this case into two subcases: $k^{*} \leq \clus^{*} - m$ and $k^{*} > \clus^{*} - m$.
  \begin{enumerate}
  \item Case (iii-a) $k^{*} \leq \clus^{*} - m$: From Case (i), we can conclude that $\clus_{j}$ in round $i^{*}$ can meet $\clus^{*}-m$ clusters. Since, $k^{*} \leq \clus^{*} - m$, $\clus^{*}$ decreases by at least $k^{*}$ in this case, thus $\Pee[S_i] = 1$.
  \item Case (iii-b) $k^{*} > \clus^{*} - m$: As shown in Case (i), $\clus_{j}$ in round $i^{*}$ can meet $\clus^{*} - m$ clusters. Moreover, it can be derived from Case (ii) that if $k^{*}_{m} \leq \lfloor \frac{\clus^{*} - m}{2} \rfloor$ heads are obtained in $\clus^{*} - m$ coin flips, $\clus^{*} - m$ decreases by $k^{*}_{m}$. Therefore, the total number of clusters decreased from $\clus^{*}$ clusters at the end of round $i^{*}$ is $\left(\clus^{*} - m\right) + k^{*}_{m}$. This number is at least $k^{*}$, because for Case (iii-b), the inequality $\left(k^{*}-k^{*}_{m}\right) \leq \left(\clus^{*} - m\right)$ holds. Therefore, $\Pee[S_i] = 1/2$.
  \end{enumerate}
We conclude the proof by showing that in all possible cases above, $\Pee[S_i] = 1/2$.
\end{proof}
It follows from Lemma~\ref{lemma:halfasynch} that the expected distance calculations in Stages 1-3 for the asynchronous case are analogous
to the synchronous calculations except in this case, we use $\alpha^{*}$ instead of $\alpha$. We summarize the calculations as follows:
\begin{enumerate}
  \item Analogous to Lemma-\ref{lemma:stage1}, the expected distance traveled during Stage-1 satisfies
\begin{align}
&\sum_{i=0}^{\alpha^{*}-1} \Ee[D_i \mid R_i] \Pee[R_i] < (r+2)\frac{n^{1.84}r^{k+6}}{r^2-1}.
\label{eq:asynchstage1}
\end{align}
  \item Analogous to Lemma-\ref{lemma:stage2}, the expected distance traveled during Stage-2 satisfies
\begin{align}
&\sum_{i=\alpha^{*}}^{\alpha^{*}+\log n-1} \Ee[D_i \mid R_i] \Pee[R_i] < (r+2)\frac{n^{2.5}r^{k+6}}{r^2-1}.
\label{eq:asynchstage2}
\end{align}
  \item Analogous to Lemma-\ref{lemma:stage3}, the expected distance traveled during Stage-3 satisfies
\begin{align}
&\sum_{i=\alpha^{*}+\log n}^{\infty} \bigg[(\Ee[ D_i \mid \notS^{*}_i]\Pee[\notS^{*}_i]) + (\Ee [ D_i \mid S^{*}_i ]\Pee[S^{*}_i])\bigg]\Pee[R_i]\notag\\
&<\frac{2n^{1.34}r^{k+6}(r+2)}{(2-r^2)(\log n)!}.\label{eq:asynchstage3}
\end{align}
\end{enumerate}

\begin{theorem}
For the choice of $r = 1.26$, $\MASR$ algorithm has a competitive ratio of $O(n^{1.5})$.
\label{theo:SRmultiasynch}
\end{theorem}
\begin{proof}
The expected distance traveled is obtained by adding the expressions in equations (\ref{eq:asynchstage1}), (\ref{eq:asynchstage2}) and (\ref{eq:asynchstage3}). We first replace each occurrence of $r^{k}$ with $dr^{-\delta}$. Then we divide by $(n-1)d$ which is the length of the optimal offline path between the clusters. This expression is maximized at $\delta = 0$. In turn, the choice of $r = 1.26$ gives the competitive ratio guarantee of $O(n^{1.5})$.
\end{proof}

\section{Simulations}
\label{sec:simulations}

To validate the performance of Algorithms $\MSR$ and $\MASR$, we ran a series
of simulations varying $n$, $r$, $d$ and the starting times. First, we present the simulation results of $\MSR$. Then, we present the simulation results that compare the performances of $\MSR$ and $\MASR$, including the effect of navigational errors.

In the left plots of Fig. \ref{simfigs:n5d10}, we investigate the performance of $\MSR$ when $n = 5$ for various $r$ values with respect to the change in $d$. Each $(d,r)$ pair is averaged over 1000 trials and $n$. The value of $r$ is varied between 1.17 and 1.3. The initial distance between each adjacent robot on the line is $2d$ which is varied between 5 and 50 step sizes. We divide the average distance traveled by $(n-1)d$ to obtain the average distance competitive ratio that is shown in the left-top plot of Fig. \ref{simfigs:n5d10}. The average distance competitive ratio remains constant for most $r$ values as $d$ changes. We observe that Algorithm $\MSR$ performs worse for small $r$ values such as $r=1.17$ and $r=1.195$. The left-middle plot in Fig. \ref{simfigs:n5d10} shows that the average distance traveled by the robots increases as $r$ decreases. We observe from left-bottom plot in Fig. \ref{simfigs:n5d10} that the average number of rounds is proportional to the average distance traveled and increases as $r$ decreases.

Next, we investigate the performance of $\MSR$ when $d = 10$ for various $r$ values with respect to the change in $n$. The results are shown in the right plots of Fig. \ref{simfigs:n5d10}. The value of $r$ is varied between 1.17 and 1.3, and $n$ is varied between 4 and 20. Each $(n,r)$ pair is averaged over 1000 trials. As expected, the average distance traveled increases as $n$ increases. The average number of rounds which is proportional to the average distance traveled also increases as $n$ increases and $r$ decreases. Recall that $d = r^{k+\delta}$. In right-bottom plot, we show the difference between the average number of rounds for rendezvous and $k/2$. For $r=1.26$, the difference is approximately $2\log n$. Thus, rendezvous occurs in $k/2+2\log n$ rounds which is earlier than the starting round $k/2+2.5\log n$ of stage-3 of the analysis.

In Fig. \ref{simfigs:nd}, we present the simulation results of $\MSR$ when both $n$ and $d$ change. We use the theoretical choice of $r = 1.26$. The values of $n$ and $d$ (in terms of step sizes) are  varied between 4 and 16. Each $(n,d)$ pair is averaged over 1000 trials. Top-left plot in Fig. \ref{simfigs:nd} show that the average distance traveled until rendezvous increases as both $n$ and $d$ increase. Top-right plot in Fig. \ref{simfigs:nd} shows that the average competitive ratio is smaller for $n > 4$, and between 8 and 9. Let ${rnd}_{first}$ denote the round a pair of robots meet for the first time and a new cluster is formed. Let ${rnd}_{total}$ be the total number of rounds required for the rendezvous. Bottom-left plot in Fig. \ref{simfigs:nd} shows the difference between ${rnd}_{total}$ and ${rnd}_{first}$ which is approximately $1.5\log n$. Let ${rnd}_{stage-3}$ be $k/2 +\alpha + \log n$ which is the starting round for Stage-3 introduced in the theoretical analysis. For each $(d,n)$ pair, we compute the difference between ${rnd}_{stage-3}$ and ${rnd}_{total}$. The results are shown in bottom-right plot in Fig.~\ref{simfigs:nd}. We observe that all values except $(n,d) = (4,4)$ pair are zero or below, implying that the rendezvous occurs at the beginning of Stage-3 or earlier rounds. This verifies the upper bound of $\MSR$ which is computed by including the expected distance traveled in Stage-3.

Finally, in Fig.\ref{simfigs:simscompar}, we compare the performances of $\MSR$, $\MASR$, and $\MASR$ in the presence of navigational errors. Robots start executing $\MASR$ $t_{j}$ time late, where $t_{j}$ is a uniform random variable generated by robot-$j$ on the interval $(0, (n-1)2d)$. In addition to the delayed start, the robots do not wait for each other to start the next round or phase of a round. In the literature, the noise caused by navigational errors is often modeled as a Gaussian whose standard deviation is proportional to the distance traveled \cite{rekleitis2004techreport}. We assume that the errors occur only in the $X$-axis while the robots are executing $\MASR$. From Tab. 1 in \cite{rekleitis2004techreport}, we use $\mu = −1.843$ and $\sigma = 0.372$. Instead of setting $f^{i} = r^{i}$ , we set $f^{i} = r^{i} + n(\mu, \sigma)$, where $n$ is the Gaussian noise. As in the theoretical analyses of $\MSR$ and $\MASR$, the comparison shown in Fig.\ref{simfigs:simscompar} suggests that robots rendezvous earlier in the synchronous case comparing to the asynchronous case. $\MASR$ with gaussian noise performs slightly better than $\MASR$.

\begin{figure}
\centering
\includegraphics[width=6cm, height=3.8cm]{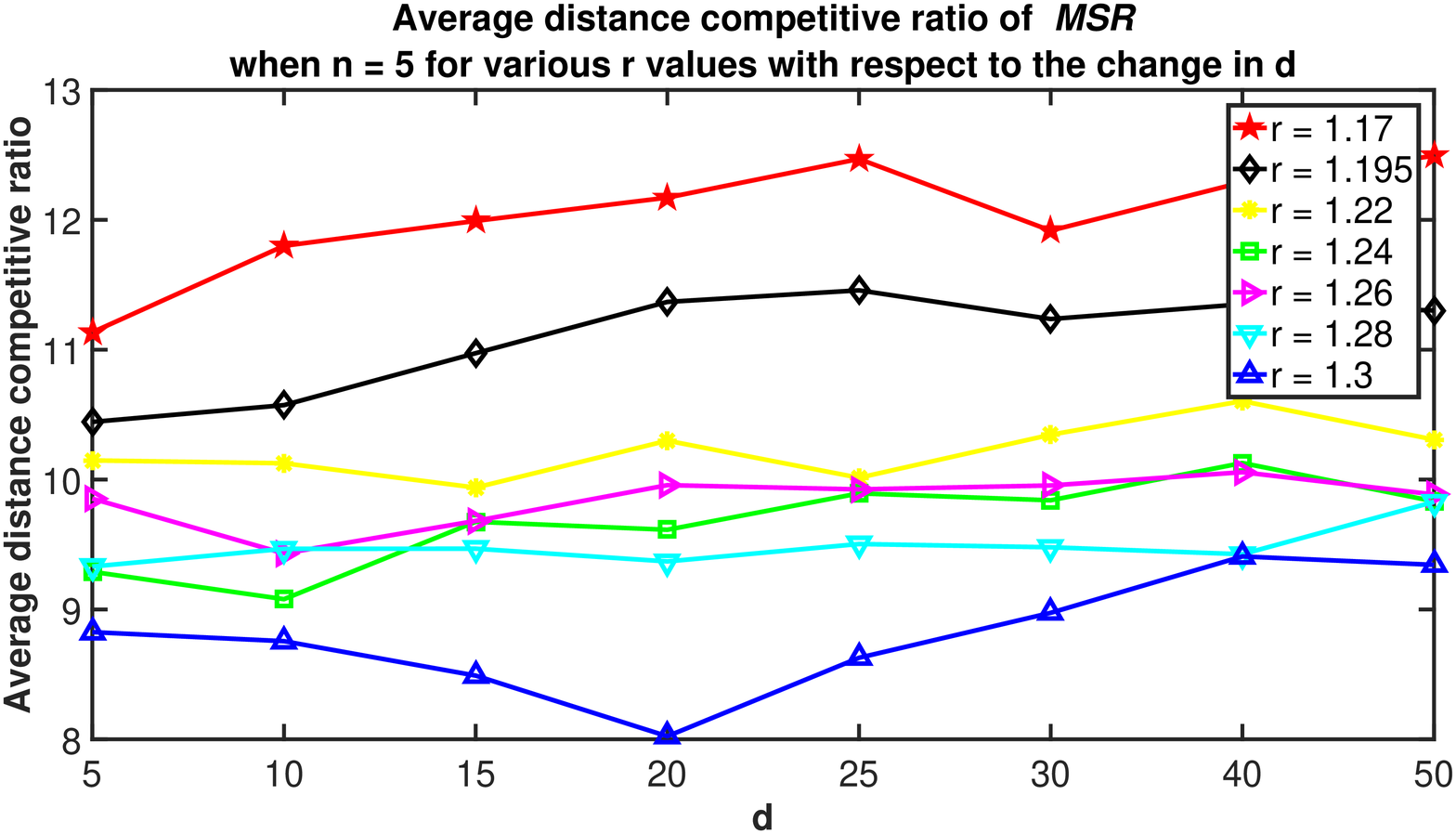}
\includegraphics[width=6cm, height=3.8cm]{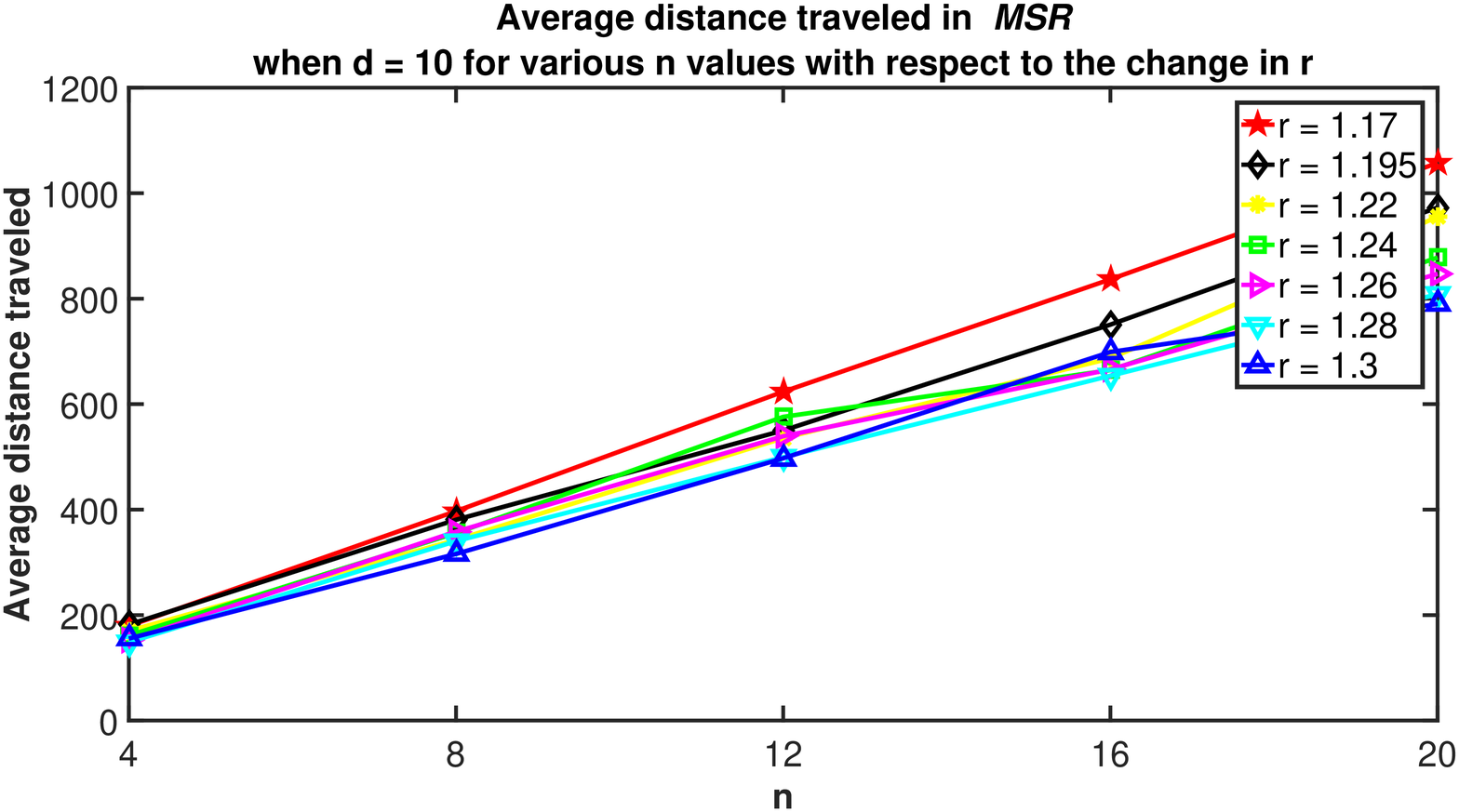}
\includegraphics[width=6cm, height=3.8cm]{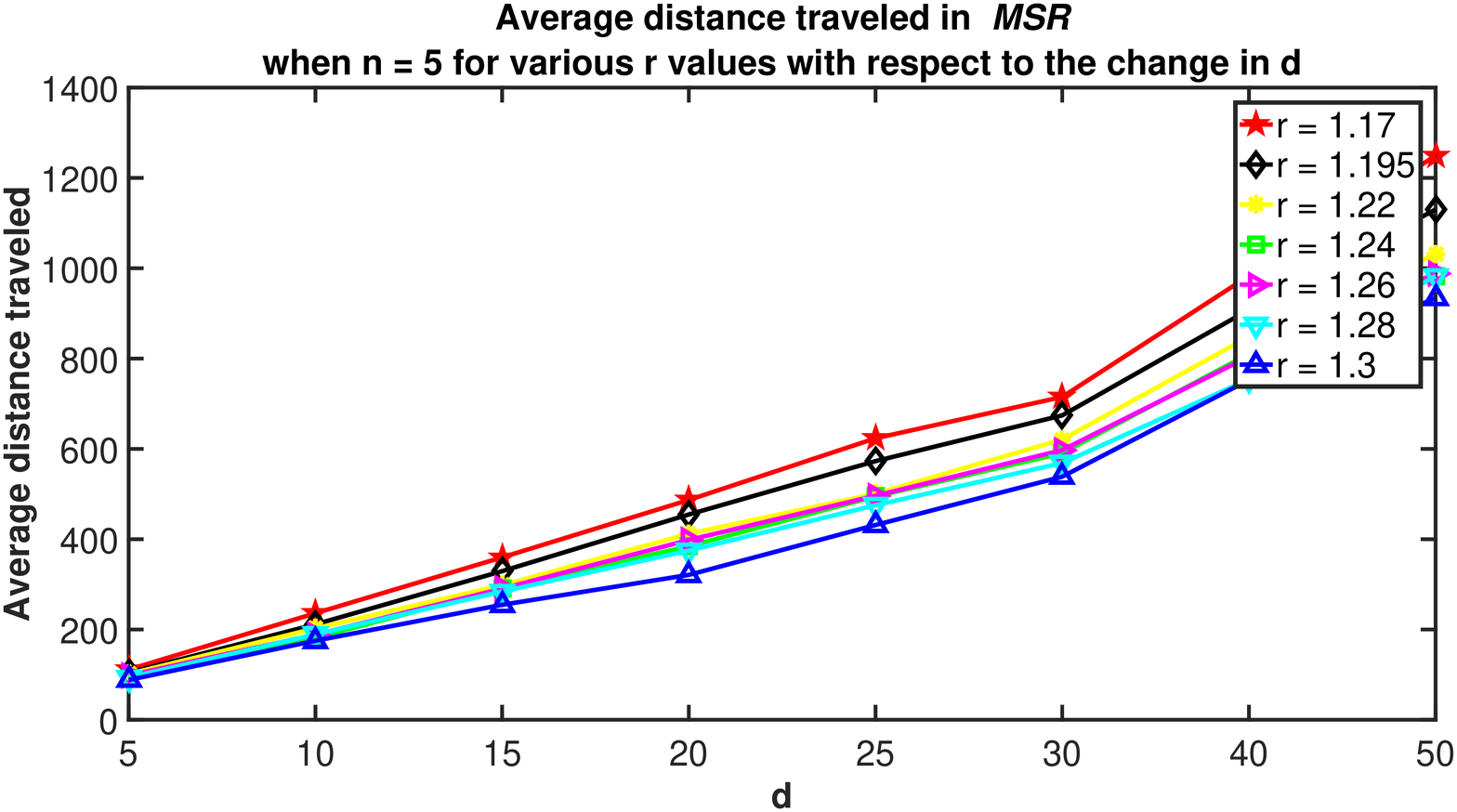}
\includegraphics[width=6cm, height=3.8cm]{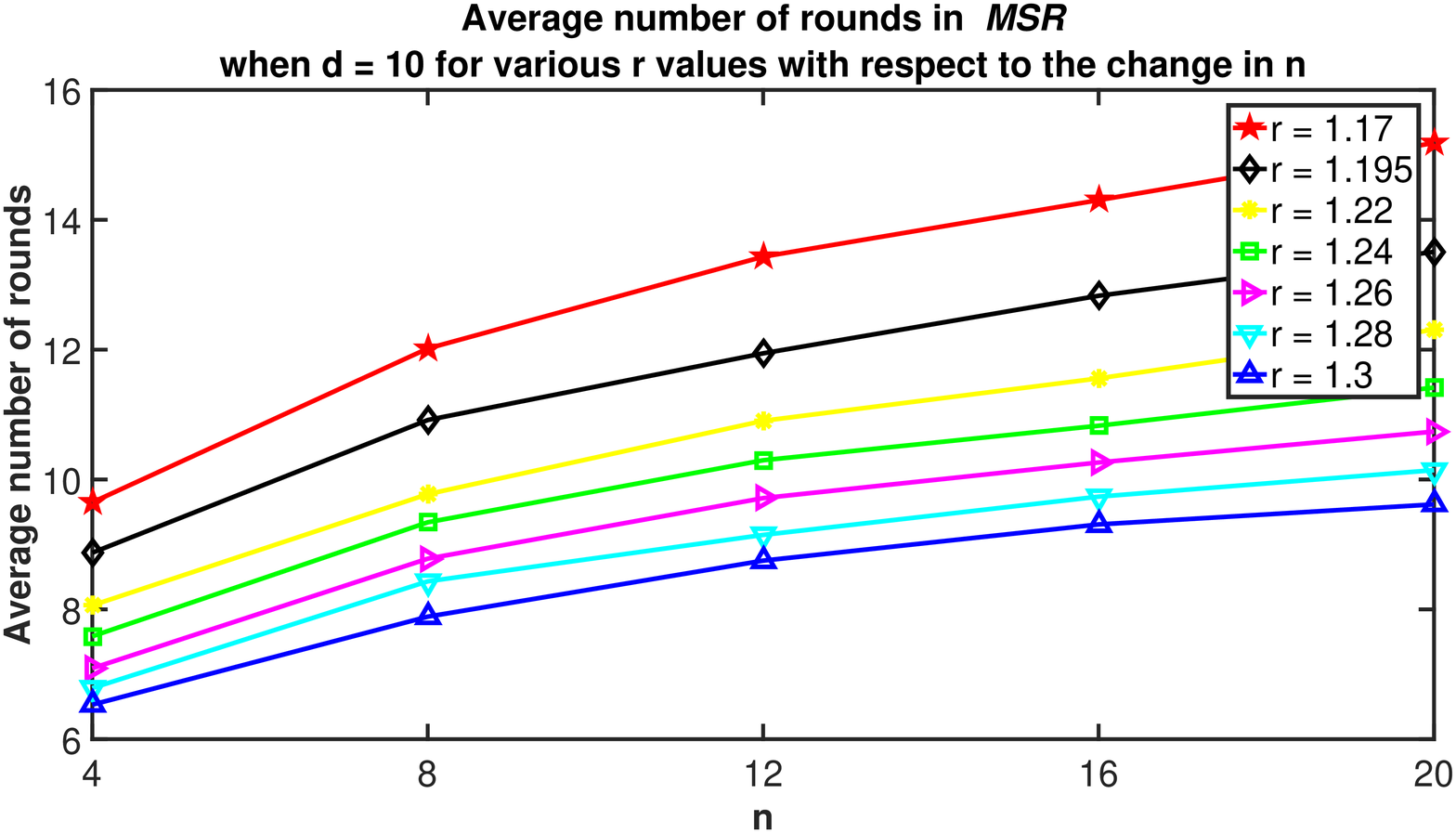}
\includegraphics[width=6cm, height=3.8cm]{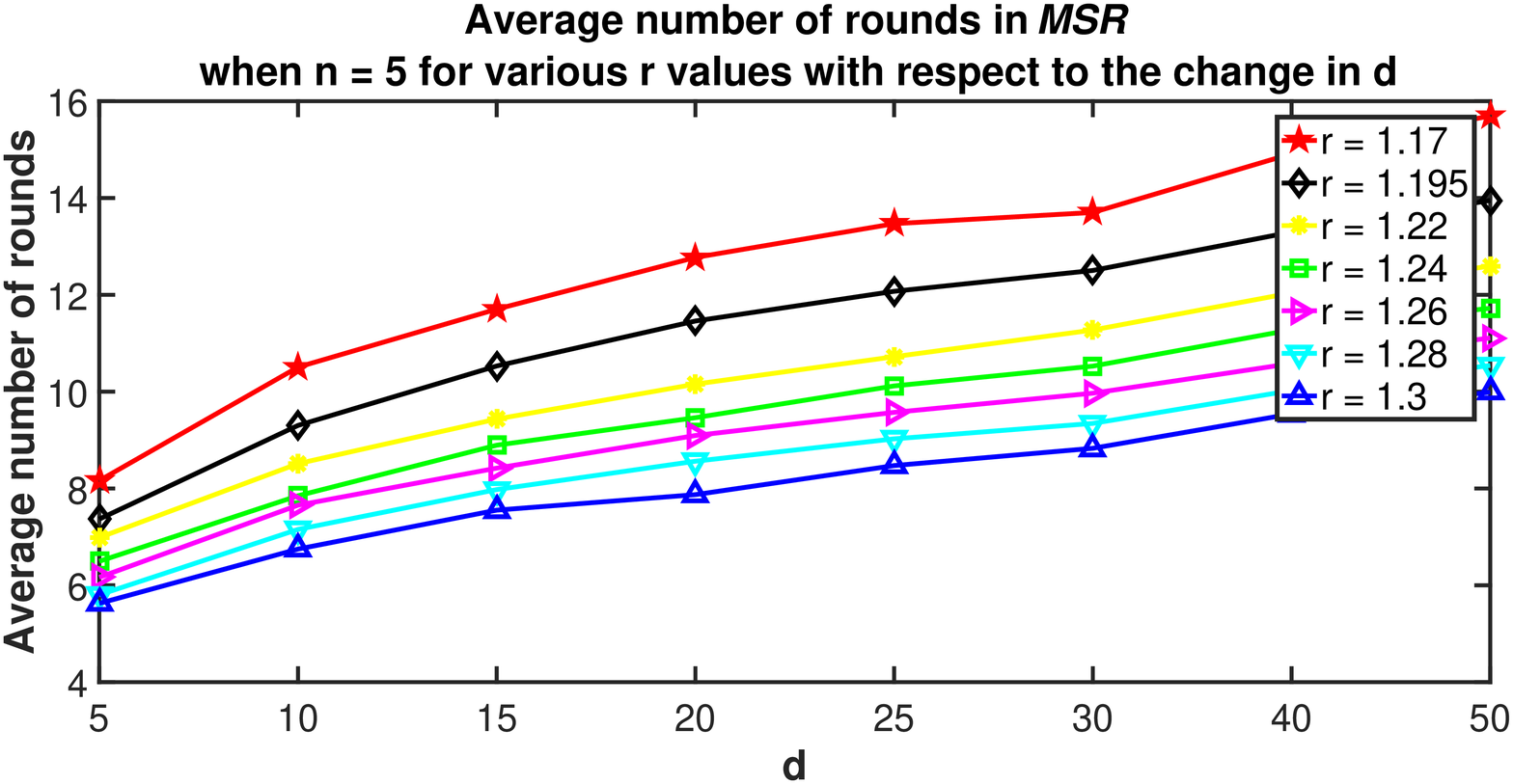}
\includegraphics[width=6cm, height=3.8cm]{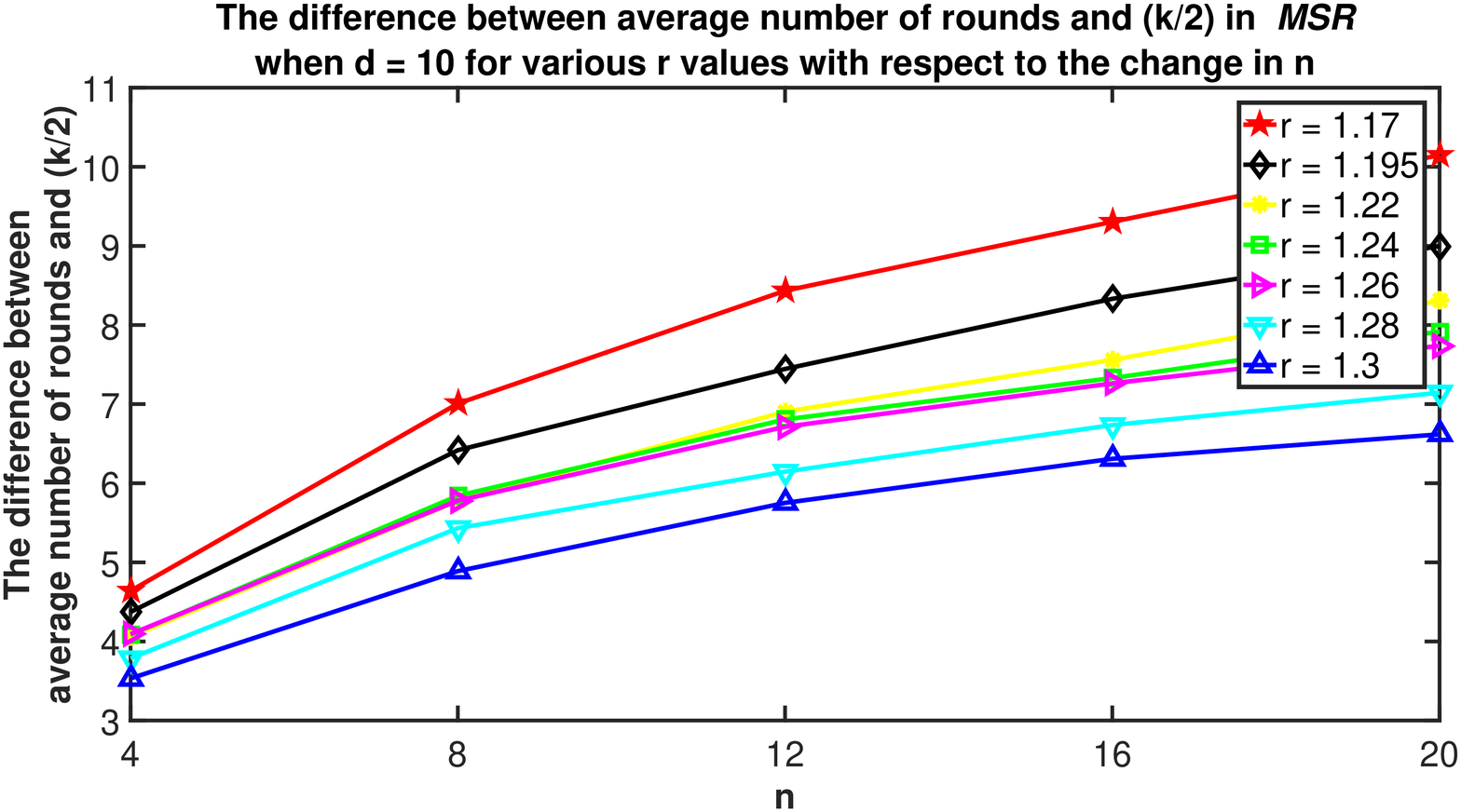}  
\caption{LEFT-plots: Simulations of the Algorithm
$\MSR$ for n=5 and various $r$ values with respect to the change in $d$. RIGHT-plots: Simulations of the Algorithm
$\MSR$ for d=10 and various $r$ values with respect to the change in $n$.}
\label{simfigs:n5d10}
\end{figure}

\begin{figure}
\centering
\includegraphics[width=6cm, height=3.8cm]{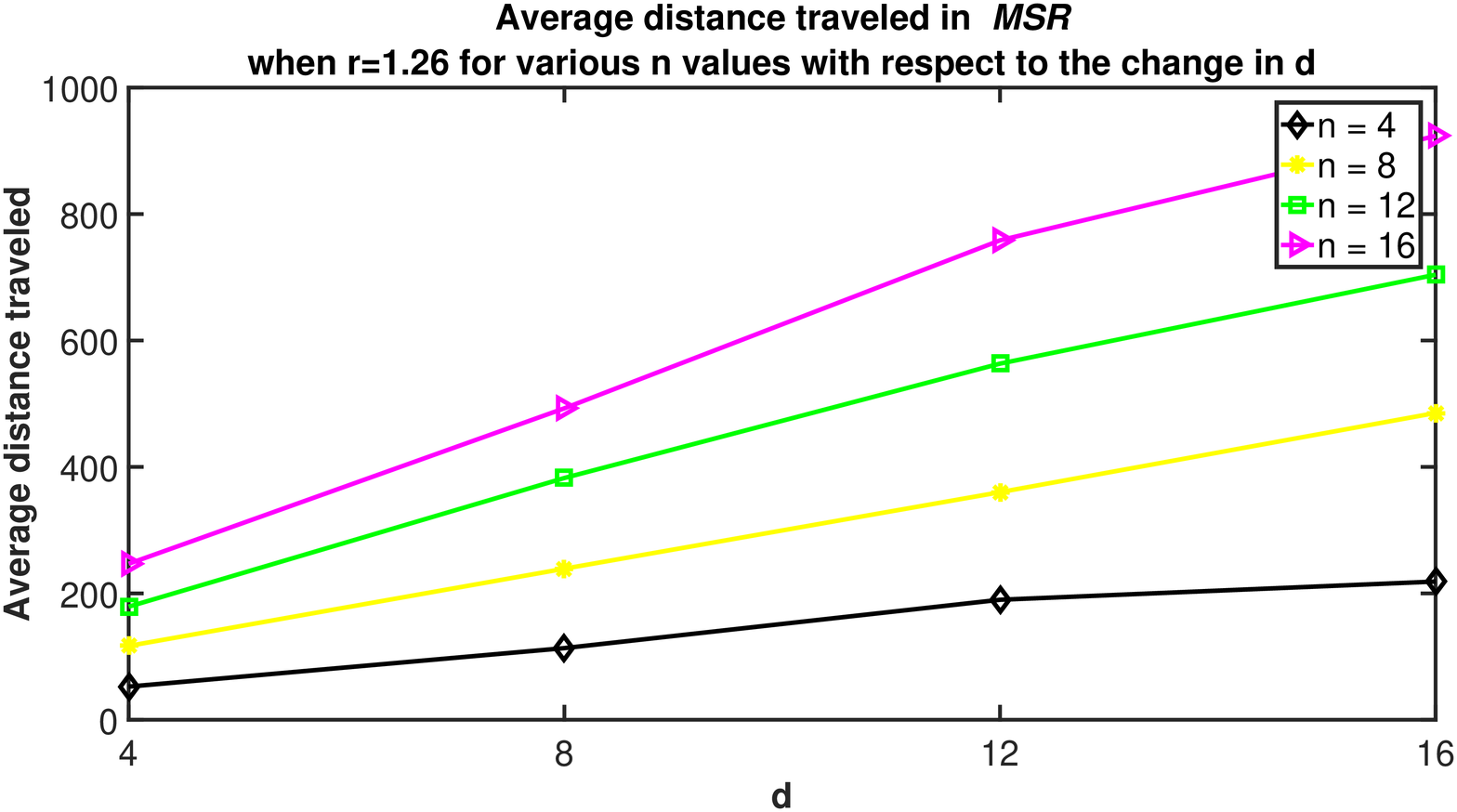}
\includegraphics[width=6cm, height=3.8cm]{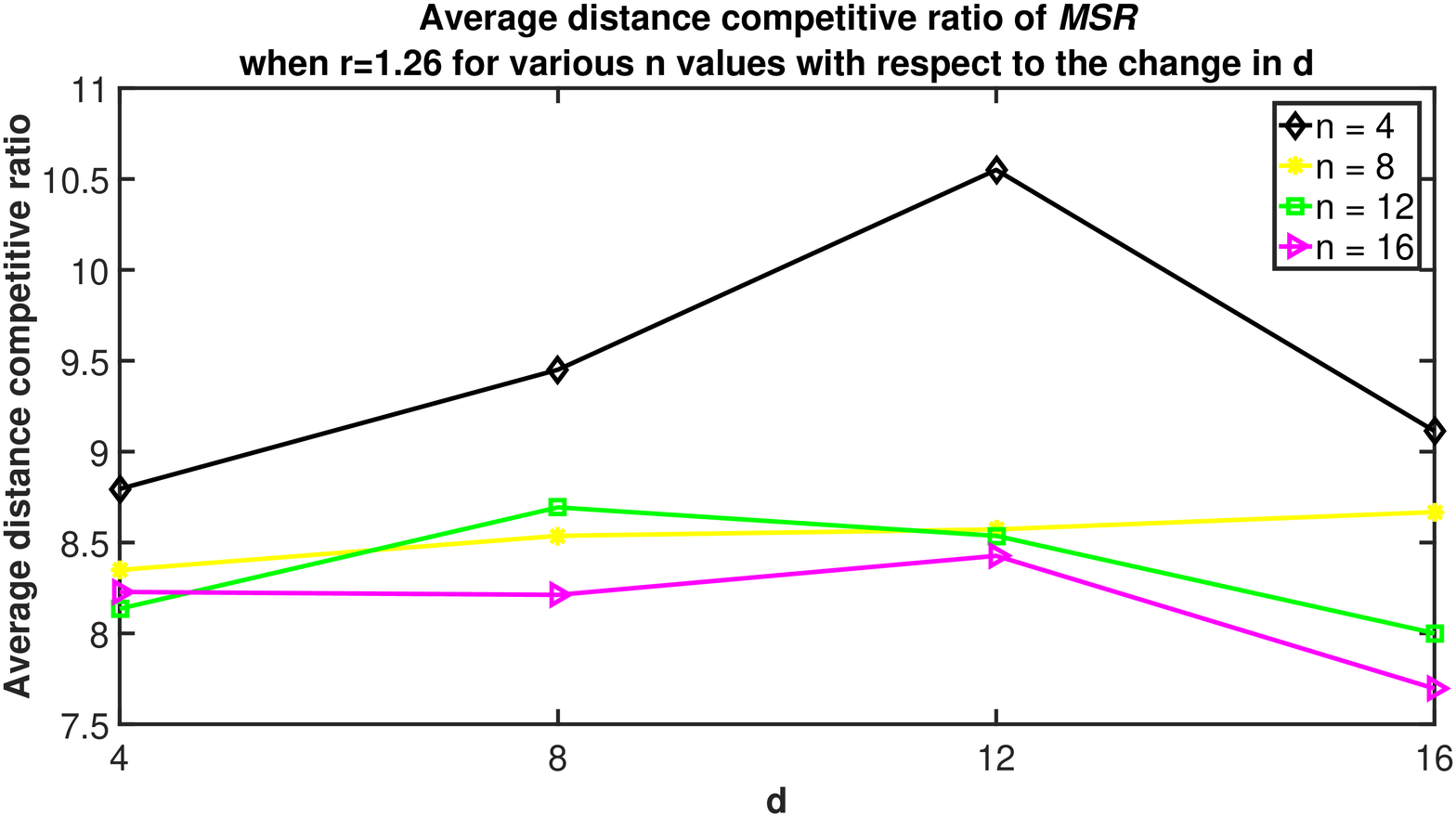}
\includegraphics[width=6cm, height=3.8cm]{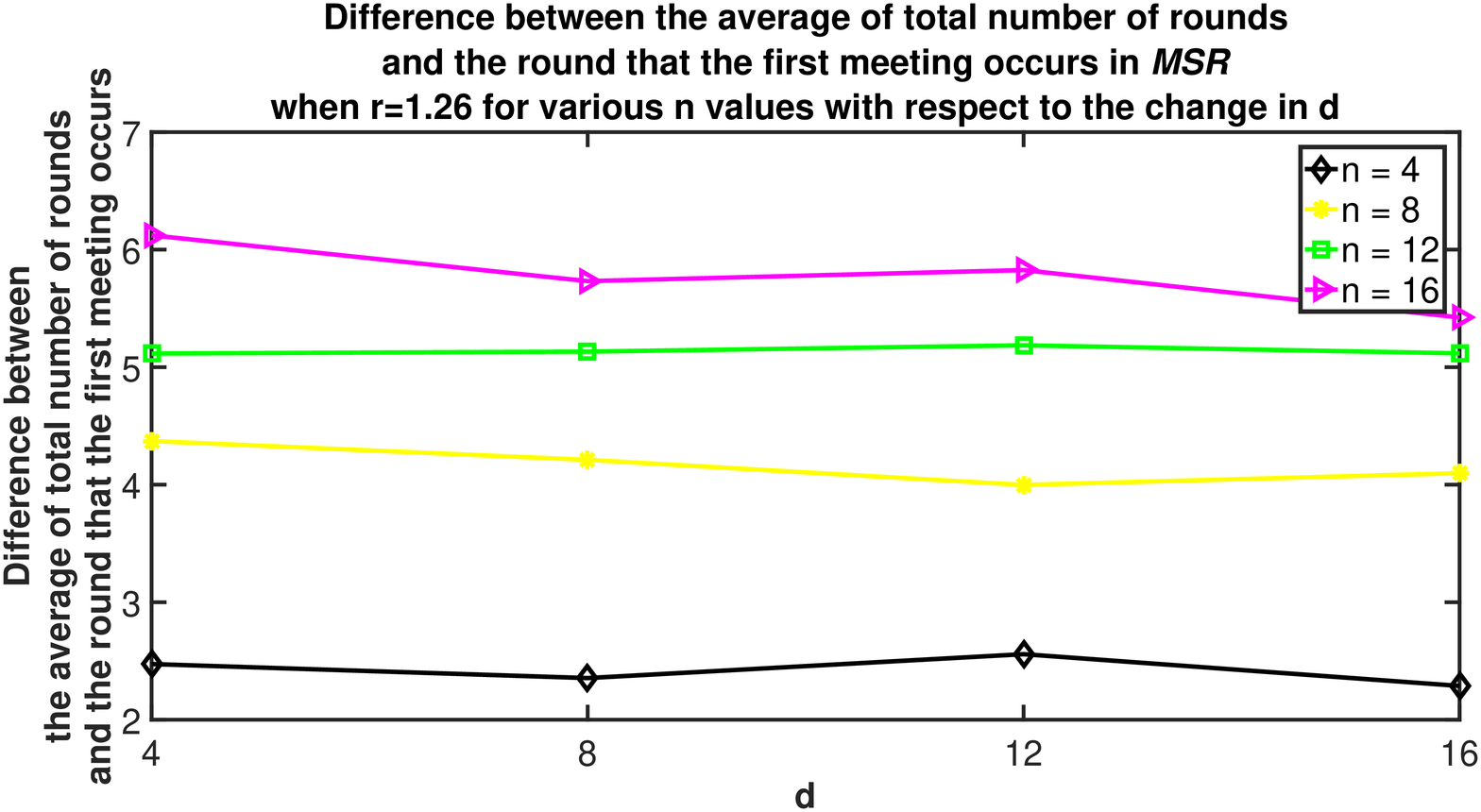}
\includegraphics[width=6cm, height=3.8cm]{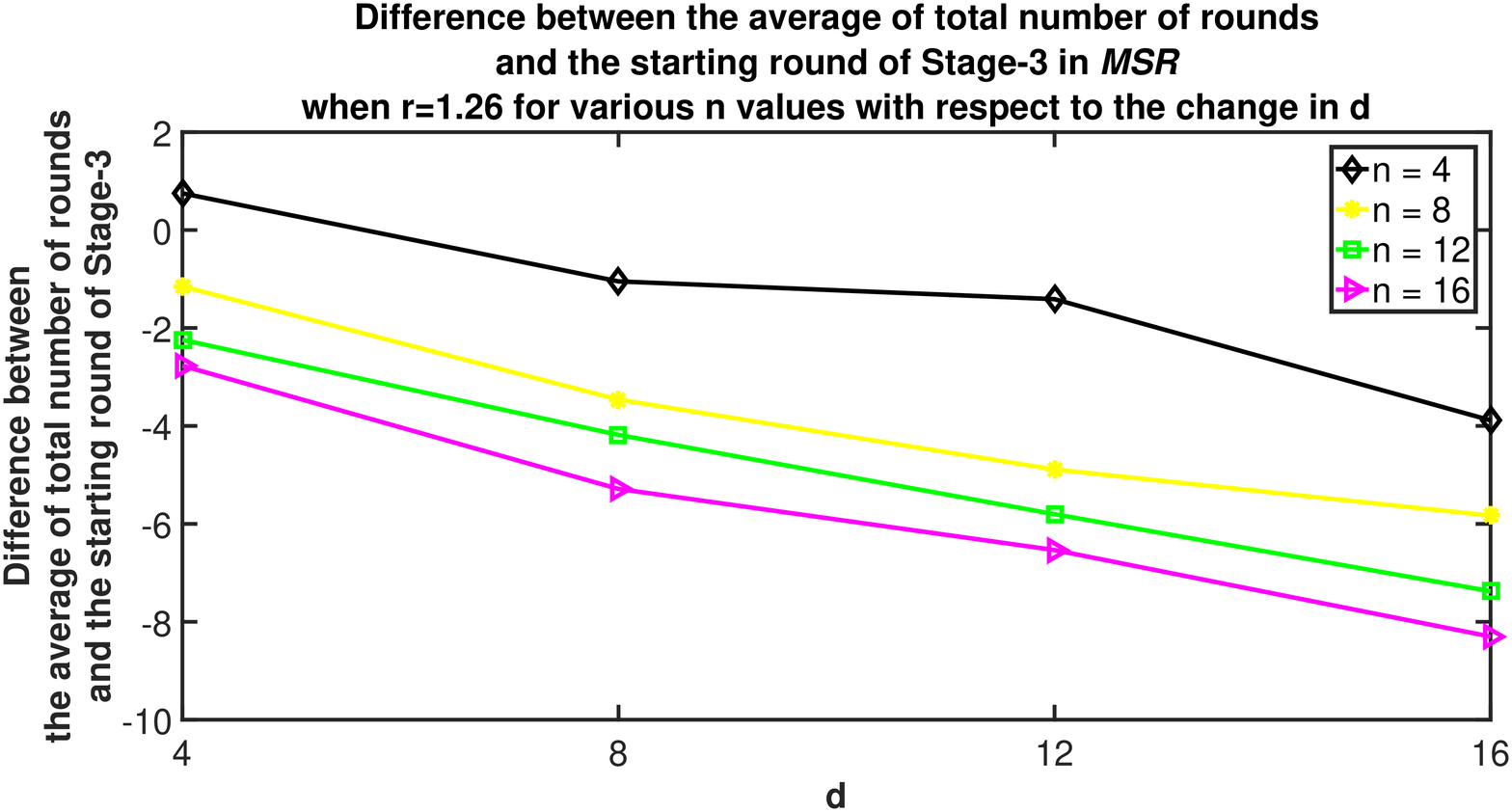}
\caption{Simulations of the Algorithm $\MSR$ for various $n$ values with respect to the change in $d$.}
\label{simfigs:nd}
\end{figure}

\begin{figure}
\centering
\includegraphics[width=6cm, height=3.8cm]{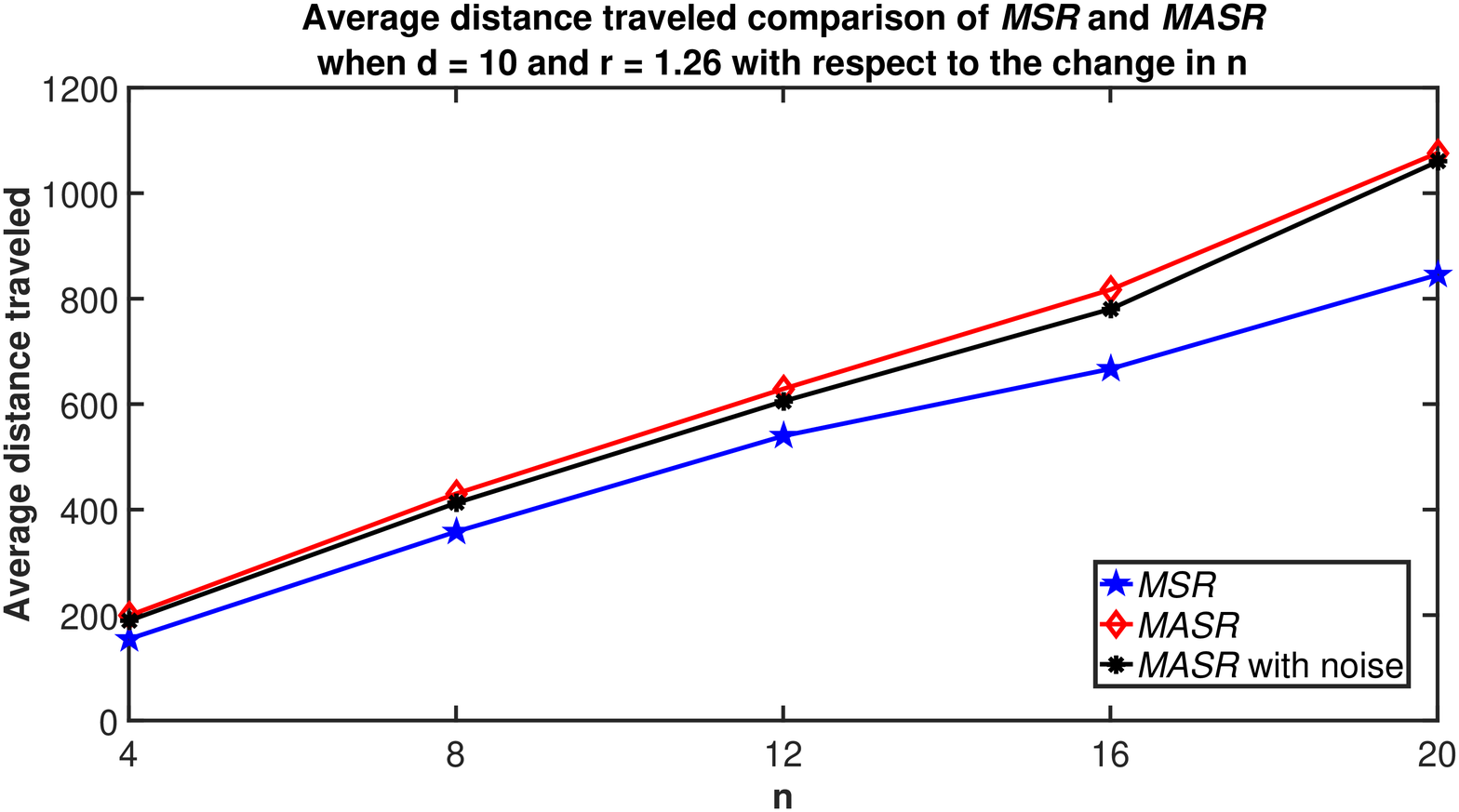}
\includegraphics[width=6cm, height=3.8cm]{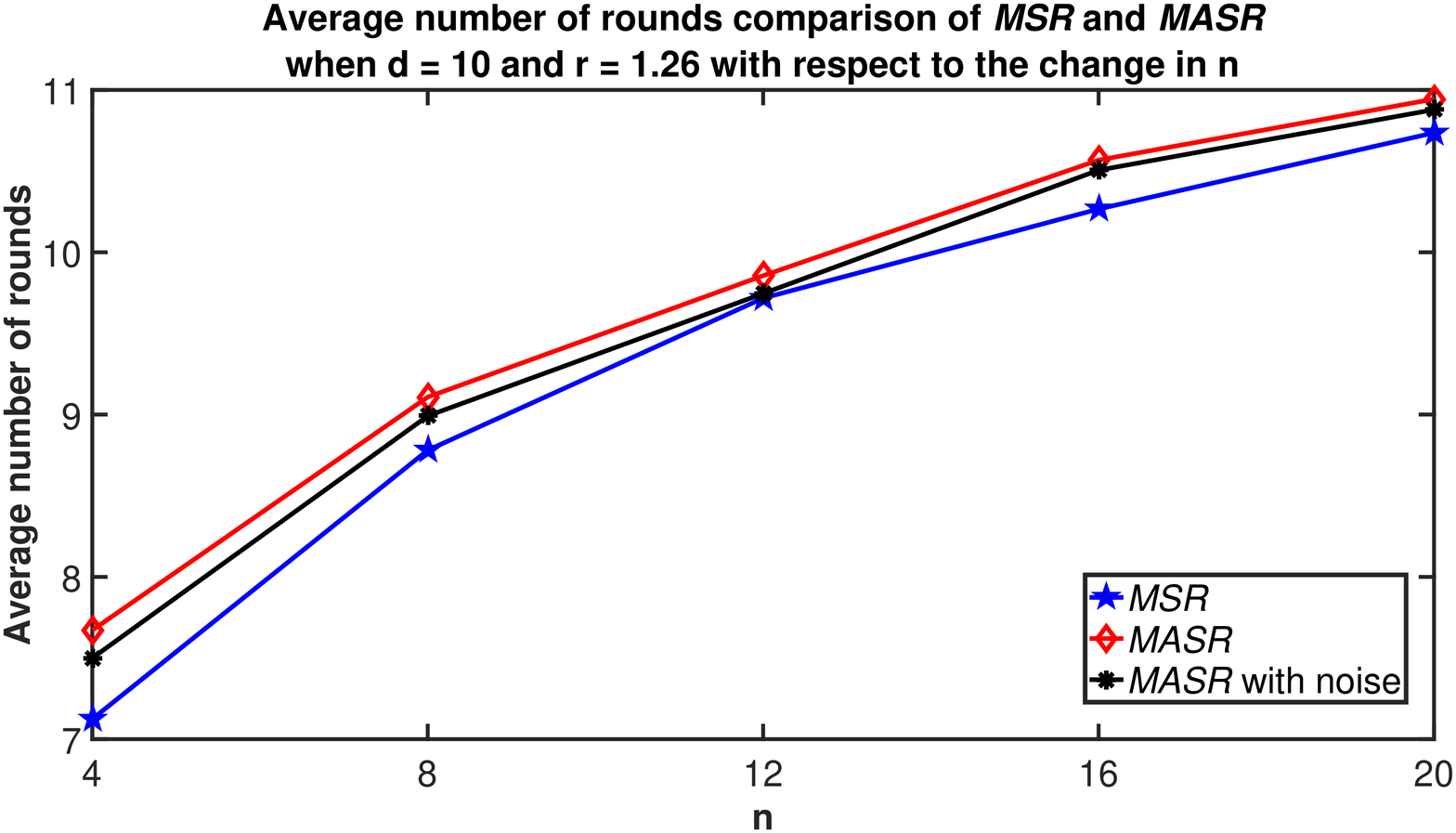}
\caption{The comparison of the performances of algorithms $\MSR$ and $\MASR$. The effect Gaussian noise on algorithm $\MASR$ is also considered. In the $\MSR$ algorithm, no waiting times are used and robot-$j$ starts executing the algorithm $t_j$ time late. $t_j$ is a uniform random variable generated by robot-$j$ on the interval $(0, (n-1)2d)$. The plots show the results with respect to the changes in $n$.}
\label{simfigs:simscompar}
\end{figure}
\section{Conclusion}
\label{sec:conclusion}

Unlike most existing work, this paper addresses the multi-robot symmetric rendezvous search problem on the line with an \emph{unknown} initial distance between the robots. We studied both the synchronous and asynchronous cases of the problem. The algorithms $\MSR$ and $\MASR$ which are proposed respectively for these cases are an extension of Algorithm $\SR$ presented in~\cite{ozsoyeller2013symmetric}. In the synchronous case, robots start executing $\MSR$ at the same time and continue to synchronize their movements in later rounds with waiting times. In the asynchronous case, the robots start executing $\MASR$ at different times. Waiting times are no longer used in $\MASR$. We showed that the competitive complexity of $\MSR$ and $\MASR$ are $O(n^{0.67})$ and $O(n^{1.5})$, respectively. Finally, we verified the theoretical bounds through simulations with respect to the change in $n$, $d$, $r$, and the starting times.

In future work, we will study the multi-robot symmetric rendezvous in graphs. The problem becomes more challenging when the robots do not know the length of the edges which can be the same or varied, also start searching at different times. 

\bibliographystyle{splncs_srt}

\end{document}